\documentclass{llncs}
\usepackage[normalem]{ulem}
\usepackage[utf8]{inputenc}
\usepackage{amsmath}
\usepackage{amsfonts}
\usepackage{amssymb}
\usepackage{latexsym}
\usepackage{gastex}
\usepackage{graphicx}
\usepackage{marvosym}
\usepackage{textcomp}
\usepackage{multirow}
\usepackage{comment}
\usepackage{wasysym}
\usepackage{array}
\newcolumntype{C}[1]{>{\centering\arraybackslash}p{#1}}
\def\abs#1{\ensuremath{\lvert #1\rvert}} 
\def\norm#1{\ensuremath{\lVert #1\rVert}}
\makeatletter
\DeclareRobustCommand\sfrac[1]{\@ifnextchar/{\@sfrac{#1}}%
                                            {\@sfrac{#1}/}}
\def\@sfrac#1/#2{\leavevmode\scalebox{.9}{\kern.1em\raise.5ex
         \hbox{$\m@th\mbox{\fontsize\sf@size\z@
                           \selectfont#1}$}\kern-.1em
         /\kern-.15em\lower.25ex
          \hbox{$\m@th\mbox{\fontsize\sf@size\z@
                            \selectfont#2}$}}}
\DeclareRobustCommand\numfrac[1]{\@ifnextchar/{\@numfrac{#1}}%
                                            {\@numfrac{#1}}}
\def\@numfrac#1{\leavevmode \hbox{$\m@th\mbox{\fontsize\sf@size\z@
                           \selectfont#1}$}}
\makeatother
\newcommand{\nat}{\mathbb N}
\newcommand{\tuple}[1]{\langle #1 \rangle}
\newcommand{\dist}{{\cal D}}

\newcommand{\q}{\hat{ q}}
\newcommand{\p}{\hat{ p}}
\newcommand{\M}{{\cal M}}
\newcommand{\N}{{\cal N}}

\newcommand{\A}{{\cal A}}

\newcommand{\Supp}{{\sf Supp}}
\newcommand{\Pref}{{\sf Pref}}
\newcommand{\Plays}{{\sf Play}}

\newcommand{\Last}{{\sf Last}}
\newcommand{\sink}{{\sf sink}}
\newcommand{\Outcomes}{\mathit{Outcomes}}
\newcommand{\fsum}{\mathit{sum}}
\newcommand{\fmax}{\mathit{max}}
\newcommand{\win}[2]{\langle \! \langle 1 \rangle \! \rangle_{\mathit{#2}}^{\mathit{#1}}}
\newcommand{\winsure}[1]{\langle \! \langle 1 \rangle \! \rangle_{\mathit{sure}}^{\mathit{#1}}}
\newcommand{\winas}[1]{\langle \! \langle 1 \rangle \! \rangle_{\mathit{almost}}^{\mathit{#1}}}
\newcommand{\winlim}[1]{\langle \! \langle 1 \rangle \! \rangle_{\mathit{limit}}^{\mathit{#1}}}

\newcommand{\sync}{{\sf sync}}
\newcommand{\init}{{\sf init}}
\newcommand{\Act}{{\sf A}}
\newcommand{\Pre}{{\sf Pre}}

\newcommand{\post}{{\sf post}}

\newcommand{\mem}{{\sf Mem}}

\let\epsilon\varepsilon

\let\oldBox\Box
\renewcommand{\Box}{\raisebox{1pt}{$\oldBox$}}
\newenvironment{exclude}{}{}

\def\mahsa#1{\marginpar{\textbf{Mahsa:} #1}}

\title{Robust Synchronization in \\ Markov Decision Processes\thanks{This work 
was partially supported by the Belgian Fonds National de la Recherche 
Scientifique (FNRS), and by the PICS project \emph{Quaverif} funded by the French 
Centre National de la Recherche Scientifique (CNRS).}}
\author{Laurent Doyen \inst{1}
\and Thierry Massart \inst{2} \and
Mahsa Shirmohammadi \inst{1,2} }

\institute{LSV, ENS Cachan \& CNRS, France 
\and Universit\'e Libre de Bruxelles, Belgium 
}
\begin{document}
\sloppy 
\maketitle 
\pagestyle{plain}

\begin{abstract}

\setlength{\parindent}{12pt}

We consider synchronizing properties of Markov decision processes (MDP),
viewed as generators of sequences of probability distributions over states.
A probability distribution is $p$-synchronizing if the probability
mass is at least $p$ in some state, and
a sequence of probability distributions is 
weakly $p$-synchronizing, or strongly $p$-synchronizing if respectively 
infinitely many, or all but finitely many distributions in the sequence 
are $p$-synchronizing.

For each synchronizing mode, an MDP can be 
$(i)$ \emph{sure} winning if there is a strategy that produces a $1$-synchronizing sequence; 
$(ii)$ \emph{almost-sure} winning if there is a strategy that produces 
a sequence that is, for all $\epsilon > 0$, a (1-$\epsilon$)-synchronizing sequence; 
$(iii)$ \emph{limit-sure} winning if for all $\epsilon > 0$, 
there is a strategy that produces a (1-$\epsilon$)-synchronizing sequence.

For each synchronizing and winning mode, 
we consider the problem of deciding whether an MDP is winning, 
and we establish matching upper and lower complexity bounds of 
the problems, as well as the optimal memory requirement
for winning strategies:  
(a) for all winning modes, we show that the problems are PSPACE-complete
for weak synchronization, and PTIME-complete for strong synchronization;
(b) we show that for weak synchronization, exponential memory is sufficient 
and may be necessary for sure winning, and infinite memory is necessary 
for almost-sure winning; for strong synchronization, linear-size memory is sufficient 
and may be necessary in all modes;
(c) we show a robustness result that the almost-sure and limit-sure winning modes coincide for
both weak and strong synchronization.
\end{abstract}


\section{Introduction}

Markov Decision Processes (MDPs) are studied in theoretical computer science 
in many problems related to system design and verification~\cite{Vardi-focs85,FV97,CY95}.
MDPs are a model of reactive systems with 
both stochastic and nondeterministic behavior, used
in the control problem for reactive systems: the nondeterminism
represents the possible choices of the controller, and the
stochasticity represents the uncertainties about the system response.
The controller synthesis problem is to compute a control strategy that
ensures correct behaviors of the system with probability~$1$.
Traditional well-studied specifications describe correct behaviors 
as infinite sequences of states, such as reachability, 
B\"uchi, and co-B\"uchi, which require the system to visit a target
state once, infinitely often, and ultimately always, respectively~\cite{AHK07,ConcOmRegGames}.

In contrast, we consider symbolic specifications of the behaviors of MDPs as sequences 
of probability distributions $X_i : Q \to [0,1]$ over the finite state
space~$Q$ of the system, where $X_i(q)$ is the probability that the MDP
is in state $q \in Q$ after $i$ steps. 
The symbolic specification of stochastic systems is relevant in  
applications such as system biology and robot planning~\cite{HMW09,BBMR08,DMS14},
and recently it has been used 
in several works on design and verification of reactive systems~\cite{KVAK10,CKVAK11,AAGT12}. 
While the verification of MDPs may yield undecidability,
both with traditional specifications~\cite{BBG08,GO10}, and symbolic
specifications~\cite{KVAK10,DMS11Err}, decidability results are obtained
for \emph{eventually synchronizing} conditions under general control  
strategies that depend on the full history of the system execution~\cite{DMS14}.
Intuitively, a sequence of probability distributions is eventually synchronizing
if the probability mass tends to accumulate in a given set of target states
along the sequence. This is an analogue, for sequences of probability 
distributions, of the reachability condition. 

In this paper, we consider an analogue of the B\"uchi 
and coB\"uchi conditions for sequences of distributions~\cite{DMS11b}: the probability
mass should get synchronized infinitely often, or ultimately at every step.
More precisely, for $0  \leq p \leq 1$ let a probability distribution
$X : Q \to [0,1]$ 
be $p$-synchronized if it assigns
probability at least~$p$ to some state. 
A sequence $\bar{X} = X_0 X_1 \dots$ of probability distributions is
$(a)$ \emph{eventually $p$-synchronizing} if $X_i$ is $p$-synchronized for some~$i$;
$(b)$ \emph{weakly $p$-synchronizing} if $X_i$ is $p$-synchronized for infinitely many~$i$'s;
$(c)$ \emph{strongly $p$-synchronizing} if $X_i$ is $p$-synchronized for all but finitely many~$i$'s.
It is easy to see that strongly $p$-synchronizing implies
weakly $p$-synchronizing, which implies eventually $p$-synchronizing.
The qualitative 
synchronizing properties, 
corresponding to the case where either $p=1$, or $p$ tends to $1$,
are analogous to the traditional 
reachability, B\"uchi, and coB\"uchi conditions.

We consider the following qualitative (winning) modes, summarized in Table~\ref{tab:def-modes}:  
$(i)$ \emph{sure} winning, if there is a strategy that generates
a \{eventually, weakly, strongly\} $1$-synchronizing sequence;
$(ii)$ \emph{almost-sure}
winning, if there is a strategy that generates a sequence that is, for all $\epsilon>0$,
\{eventually, weakly, strongly\} $(1-\epsilon)$-synchronizing;
$(iii)$ \emph{limit-sure} winning, if for all $\epsilon>0$, there is a strategy
that generates a \{eventually, weakly, strongly\} $(1-\epsilon)$-synchronizing sequence.

For eventually synchronizing deciding if a given MDP is winning is PSPACE-complete, 
and the three winning modes form a strict hierarchy~\cite{DMS14}. 
In particular, there are limit-sure winning MDPs that are not almost-sure winning. 
An important and difficult result in this paper is that the new synchronizing modes are more robust: 
for weak and strong synchronization, we show that the almost-sure and limit-sure modes coincide.
Moreover we establish the complexity of deciding if a given MDP is winning
by providing tight (matching) upper and lower bounds: 
for each winning mode we show that the problems are PSPACE-complete 
for weak synchronization, and PTIME-complete for strong synchronization.  

Thus the weakly and strongly
synchronizing properties provide conservative approximations of eventually synchronizing,
they are robust (limit-sure and almost-sure coincide), and they are of the same
(or even lower) complexity as compared to eventually synchronizing.

\begin{table}[!t]
\begin{center}
\scalebox{0.85}{
\begin{tabular}{|l@{\ } |*{9}{c@{\;}c@{\;}c@{\;}|}}
\hline                        
 \large{\strut}   & \multicolumn{3}{|c|}{Eventually} & \multicolumn{3}{|c|}{Weakly} & \multicolumn{3}{|c|}{Strongly} \\
\hline
 Sure \large{\strut}	  
        	        & $\exists \alpha$ & $\exists n$                    & $\M^{\alpha}_n(T)=1$ 
			& $\exists \alpha$ & $\forall N \: \exists n \geq N$   & $\M^{\alpha}_n(T)=1$ 
			& $\exists \alpha$ & $\exists N \: \forall n \geq N$   & $\M^{\alpha}_n(T)=1$ 
 \\
\hline
 Almost-sure \  \large{\strut}	  
				& $\exists \alpha$& $\sup_{n}$                 & $\M^{\alpha}_n(T)=1$ 
				& $\exists \alpha$& $\limsup_{n\to\infty}$     & $\M^{\alpha}_n(T)=1$ 
				& $\exists \alpha$& $\liminf_{n\to\infty}$     & $\M^{\alpha}_n(T)=1$ 
 \\
\hline
 Limit-sure \large{\strut}      
		& $\sup_{\alpha}$& $\sup_{n}$                & $\M^{\alpha}_n(T)=1$  
		& $\sup_{\alpha}$& $\limsup_{n\to\infty}$    & $\M^{\alpha}_n(T)=1$  
		& $\sup_{\alpha}$& $\liminf_{n\to\infty}$    & $\M^{\alpha}_n(T)=1$  
\\
\hline
\end{tabular}  
}
\end{center}
\caption{Winning modes and synchronizing objectives 
(where $\M^{\alpha}_n(T)$ denotes the probability that under strategy $\alpha$,
after $n$ steps the MDP $\M$ is in a state of $T$). \label{tab:def-modes}}
\end{table}

We also provide optimal memory bounds for 
winning strategies:  
exponential memory is sufficient and may be necessary
for sure winning in weak synchronization,
infinite memory is necessary for almost-sure winning in weak synchronization,
and linear memory is sufficient for strong synchronization in all 
winning modes. We present a variant of strong synchronization
for which memoryless strategies are sufficient.

\paragraph{Related Works and Applications.}
Synchronization problems were first considered for deterministic finite automata (DFA)
where a \emph{synchronizing word} is a finite sequence of control actions that can be 
executed from any state of an automaton and leads to the same state 
(see~\cite{Volkov08} for a survey of results and applications). 
While the existence of a synchronizing word can be decided in polynomial 
time for DFA, extensive research efforts are devoted to establishing a tight bound
on the length of the shortest synchronizing word, which is
conjectured to be $(n-1)^2$ for automata with $n$ states~\cite{Cer64}.
Various extensions of the notion of synchronizing word have been proposed 
for non-deterministic and probabilistic automata~\cite{Burkhard76a,IS99,Kfo70,DMS11a},
leading to results of PSPACE-completeness~\cite{Martyugin14}, or even 
undecidability~\cite{Kfo70,DMS11Err}. 


For probabilistic systems, a natural extension of words
is the notion of strategy that reacts and chooses actions according to the sequence 
of states visited along the system execution. In this context, an input word 
corresponds to the special case of a blind strategy that chooses the control actions
in advance. 
In particular, almost-sure weak and strong synchronization 
with blind strategies has been studied~\cite{DMS11a} and the main result is the undecidability
of deciding the existence of a blind almost-sure winning strategy for
weak synchronization, and the PSPACE-completeness of the emptiness
problem for strong synchronization~\cite{DMS11b,DMS11Err}.
In contrast, for general strategies (which also correspond to input trees),
we establish the PSPACE-completeness and 
PTIME-completeness of deciding almost-sure weak and strong 
synchronization respectively. 

A typical application scenario is the design of a control program for 
a group of mobile robots running in a stochastic environment. 
The possible behaviors of the robots and the stochastic response of 
the environment (such as obstacle encounters) are represented by an MDP,
and a synchronizing strategy corresponds to a control program that
can be embedded in every robot to ensure that they meet (or synchronize)
eventually once, infinitely often, or eventually forever.

\section{Markov Decision Processes and Synchronization}\label{sec:def}

We closely follow the definitions of~\cite{DMS14}.
A \emph{probability distribution} over a finite set~$S$ is a
function $d : S \to [0, 1]$ such that $\sum_{s \in S} d(s)= 1$. 
The \emph{support} of~$d$ is the set $\Supp(d) = \{s \in S \mid d(s) > 0\}$. 
We denote by $\dist(S)$ the set of all probability distributions over~$S$. 
Given a set $T\subseteq S$, let $d(T)=\sum_{s\in T}d(s)$
and $\norm{d}_T = \fmax_{s\in T} d(s)$. 
For $T \neq \emptyset$, the \emph{uniform distribution} on $T$ assigns probability 
$\frac{1}{\abs{T}}$ to every state in $T$.
Given $s \in S$, the \emph{Dirac distribution} on~$s$ assigns probability~$1$
to~$s$, and by a slight abuse of notation, we denote it simply by $s$.


A \emph{Markov decision process} (MDP) is a tuple $\M = \tuple{Q,\Act,\delta}$ 
where $Q$ is a finite set of states, $\Act$ is a finite set of actions, 
and $\delta: Q \times \Act \to \dist(Q)$ is a probabilistic transition function.   
A state $q$ is \emph{absorbing} if $\delta(q,a)$ is the Dirac distribution on $q$
for all actions $a \in \Act$. 


Given state $q\in Q$ and action $a \in \Act$, the successor state of $q$
under action $a$ is $q'$ with probability $\delta(q,a)(q')$.
Denote by $\post(q,a)$ the set $\Supp(\delta(q,a))$, 
and given $T \subseteq Q$ let $\Pre(T)= \{q \in
Q \mid \exists a \in \Act: \post(q, a) \subseteq T\}$ be the set of
states from which there is an action to ensure that the
successor state is in $T$.  For $k>0$, let $\Pre^k(T) =
\Pre(\Pre^{k-1}(T))$ with $\Pre^0(T)=T$.

A \emph{path} in $\M$ is an infinite sequence $\pi = q_{0} a_{0} q_{1} a_1 \dots$
such that $q_{i+1} \in \post(q_{i},a_{i})$ 
for all $i \geq 0$. A finite prefix $\rho = q_{0} a_{0} q_{1} a_1 \dots q_n$
of a path (or simply
a finite path) has length $\abs{\rho}=n$ and last state $\Last(\rho)=q_{n}$.
We denote by $\Plays(\M)$ and $\Pref(\M)$ the set of all
paths and finite paths in~$\M$ respectively.


\paragraph{Strategies.}
A \textit{randomized strategy} for $\M$ (or simply a strategy) is a function 
$\alpha: \Pref(\M) \to \dist(\Act)$ that, given a finite path $\rho$,
returns a probability distribution $\alpha(\rho)$ over the action set,
used to select a successor state $q'$ of $\rho$ with probability 
$\sum_{a \in \Act} \alpha(\rho)(a) \cdot \delta(q,a)(q')$ where $q=\Last(\rho)$.

A strategy $\alpha$ is
\emph{pure} if for all $\rho \in \Pref(\M)$, 
there exists an action $a \in \Act$ such that $\alpha(\rho)(a)=1$; and 
\emph{memoryless} if $\alpha(\rho) = \alpha(\rho')$
for all $\rho, \rho'$ such that $\Last(\rho) = \Last(\rho')$.
We view pure strategies as functions $\alpha: \Pref(\M)\to \Act$, and 
memoryless strategies as functions $\alpha: Q \to \dist(\Act)$.


Finally, a strategy $\alpha$ uses \emph{finite-memory} if it can be represented 
by a finite-state transducer $T = \tuple{\mem,m_0, \alpha_u, \alpha_n}$ where
$\mem$ is a finite set of modes (the memory of the strategy), 
$m_0 \in \mem$ is the initial mode,
$\alpha_u: \mem \times (\Act \times Q) \to \mem$ is an update function, 
that given the current memory, last action and state updates the memory,
and $\alpha_n: \mem \times Q \to \dist(\Act)$ is a next-move function
that selects the probability distribution $\alpha_n(m,q)$ over
actions when the current mode is $m$ and the current state of $\M$ is $q$.
For pure strategies, we assume that $\alpha_n: \mem \times Q \to \Act$.
The \emph{memory size} of the strategy is the number $\abs{\mem}$ of modes.
For a finite-memory strategy $\alpha$, let $\M(\alpha)$ be the Markov chain
obtained as the product of $\M$ with the transducer defining $\alpha$.
We assume general knowledge of the reader about Markov chains, 
such as recurrent and transient states, periodicity, and 
stationary distributions~\cite{PK11}.


\paragraph{Outcomes and winning modes.}
Given an initial distribution $d_0 \in \dist(Q)$ and a strategy $\alpha$ in an MDP $\M$,
a \emph{path-outcome} is a path
$\pi = q_{0}  a_{0} q_{1} a_1 \dots$ in $\M$
such that $q_0 \in \Supp(d_0)$ and $a_{i} \in \Supp(\alpha(q_0 a_0 \dots q_{i}))$ 
for all $i \geq 0$. The probability of a finite prefix 
$\rho = q_{0} a_{0} q_{1} a_1 \dots q_n$ of $\pi$ is 
$$d_0(q_{0}) \cdot
\prod_{j=0}^{n-1} \alpha(q_{0} a_{0}\dots q_{j})(a_{j}) \cdot
\delta(q_{j},a_{j})(q_{j+1}).$$
We denote by $\Outcomes(d_0,\alpha)$ the set of all path-outcomes from $d_0$ under strategy $\alpha$.
An \emph{event} $\Omega \subseteq \Plays(\M)$ is a measurable set of paths, and given an initial distribution 
$d_0$ and a strategy $\alpha$, the probability $Pr^{\alpha}(\Omega)$ of
$\Omega$ is uniquely defined~\cite{Vardi-focs85}.
We consider the following classical winning modes. 
Given an initial distribution $d_0$ and an event $\Omega$, we say that $\M$ is:
\emph{sure winning} if there exists a strategy $\alpha$ such that $\Outcomes(d_0,\alpha) \subseteq \Omega$;
\emph{almost-sure winning} if there exists a strategy $\alpha$ such that $\Pr^{\alpha}(\Omega) = 1$;

For example, given a set $T \subseteq Q$ of target states, and $k \in \nat$, 
we denote by 
$\Box T = \{q_{0} a_{0} q_{1} \dots \in \Plays(\M) \mid \forall i: q_i \in T\}$
the safety event of always staying in~$T$, by
$\Diamond T = \{q_{0} a_{0} q_{1} \dots \in \Plays(\M) \mid \exists i: q_i \in T\}$
the event of reaching~$T$, and by
$\Diamond^{k}\, T = \{q_{0} a_{0} q_{1} \dots \in \Plays(\M) \mid q_k \in T \}$ the event
of reaching~$T$ after exactly $k$~steps. Let $\Diamond^{\leq k}\, T = \bigcup_{j \leq k} \Diamond^{j}\, T$.
Hence, if $\Pr^{\alpha}(\Diamond T)=1$ then almost-surely  a state in~$T$ is reached 
under strategy $\alpha$.

\begin{table}[t]
\begin{center}
\scalebox{0.90}{
\begin{tabular}{|l@{\ }|c|c|c|}
\hline                        
\large{\strut}           &  Eventually & Weakly   & Strongly \\
\hline                        
Sure \large{\strut}      &  \;PSPACE-C\;   & \;{\bf PSPACE-C}\; & \;{\bf PTIME-C}\;  \\
\hline                        
Almost-sure \large{\strut} & PSPACE-C   & \multirow{2}{*}{{\bf PSPACE-C}} & \multirow{2}{*}{{\bf PTIME-C}}  \\
\cline{1-2}
Limit-sure \large{\strut}  & PSPACE-C   &  &   \\
\hline
\end{tabular}  
}
\end{center}
\caption{Computational complexity of the membership problem (new results in boldface).\label{tab:complexity}}
\end{table}


We consider a symbolic outcome of MDPs viewed as generators of sequences of 
probability distributions over states~\cite{KVAK10}.
Given an initial distribution $d_0 \in \dist(Q)$ and a strategy $\alpha$ in $\M$,
the \emph{symbolic outcome} of $\M$ from $d_0$ is 
the sequence $(\M^{\alpha}_n)_{n\in \nat}$ of probability distributions defined by 
$\M^{\alpha}_k(q) = Pr^{\alpha}(\Diamond^{k}\, \{q\})$ for all $k \geq 0$ and $q \in Q$.
Hence, $\M^{\alpha}_k$ is the probability distribution over states after $k$ steps
under strategy $\alpha$.
Note that $\M^{\alpha}_0 = d_0$ and the symbolic outcome is a deterministic sequence 
of distributions: each distribution $\M^{\alpha}_k$ has a unique (determinisitic) 
successor.

Informally, synchronizing objectives require that the probability of some state 
(or some group of states) tends to $1$ in the sequence $(\M^{\alpha}_n)_{n\in \nat}$,
either once, infinitely often, or always after some point. 
Given a set $T \subseteq Q$, consider the functions $\fsum_T: \dist(Q) \to [0,1]$ 
and $\fmax_T: \dist(Q) \to [0,1]$ that compute $\fsum_T(X) = \sum_{q \in T} X(q)$
and $\fmax_T(X) = \max_{q \in T} X(q)$. 
For $f \in \{\fsum_T,\fmax_T\}$ and $p \in [0,1]$, 
we say that a probability distribution $X$ is $p$-synchronized according to $f$ if $f(X) \geq p$,
and that a sequence $\bar{X} = X_0 X_1 \dots$ of
probability distributions is~\cite{DMS11b,DMS14}:
\begin{itemize}
\item[$(a)$] \emph{event} (or \emph{eventually}) \emph{$p$-synchronizing} if $X_i$ is $p$-synchronized for some $i \geq 0$;
\item[$(b)$] \emph{weakly $p$-synchronizing} if $X_i$ is $p$-synchronized for infinitely many $i$'s;
\item[$(c)$] \emph{strongly $p$-synchronizing} if $X_i$ is $p$-synchronized for all but finitely many $i$'s.
\end{itemize}

For $p=1$, these definitions are analogous to the traditional reachability, 
B\"uchi, and coB\"uchi conditions~\cite{ConcOmRegGames}, and the following 
winning modes can be considered~\cite{DMS14}:
given an initial distribution $d_0$ and a function $f \in \{\fsum_T,\fmax_T\}$, 
we say that for the objective of \{eventually, weak, strong\} synchronization
from~$d_0$, $\M$ is:
\begin{itemize}
\item \emph{sure winning} if there exists a strategy $\alpha$ such that
the symbolic outcome of $\alpha$ from $d_0$ 
is \{eventually, weakly, strongly\} $1$-synchronizing according to~$f$;
\item \emph{almost-sure winning} if there exists a strategy $\alpha$ such that 
for all $\epsilon>0$ the symbolic outcome of $\alpha$ from $d_0$ is 
\{eventually, weakly, strongly\} $(1-\epsilon)$-synchronizing according to $f$;
\item \emph{limit-sure winning} if for all $\epsilon>0$, there exists a strategy $\alpha$
such that the symbolic outcome of $\alpha$ from $d_0$ is 
\{eventually, weakly, strongly\} $(1-\epsilon)$-synchronizing according to $f$;
\end{itemize}

Note that the winning modes for synchronization objectives differ
from the classical winning modes in MDPs: they can be viewed as a 
specification of the set of sequences of distributions that are 
winning in a non-stochastic system (since the symbolic outcome is
deterministic), while the traditional almost-sure and limit-sure winning modes 
for path-outcomes consider a probability measure over paths and specify 
the probability of a specific event (i.e., a set of paths).
Thus for instance a strategy is almost-sure synchronizing 
if the (single) symbolic outcome it produces belongs to the corresponding
winning set, whereas traditional almost-sure winning requires a certain event
to occur with probability $1$.

We often write $\norm{X}_T$ instead
of  $\fmax_T(X)$ (and we omit the subscript when $T=Q$) 
and $X(T)$ instead of $\fsum_T(X)$, as 
in Table~\ref{tab:def-modes} where 
the definitions of the various winning modes and synchronizing objectives
for $f=\fsum_T$ 
are summarized. 

\begin{table}[t]
\begin{center}
\scalebox{0.90}{
\begin{tabular}{|l@{\ }|c|c|C{22mm}|C{22mm}|}
\hline                        
\large{\strut}           &  \multirow{2}{*}{Eventually} & \multirow{2}{*}{Weakly}   & \multicolumn{2}{c|}{Strongly} \\
\cline{4-5}
\large{\strut}           &   &    & $\fsum_T$ & $\fmax_T$ \\
\hline                        
Sure \large{\strut}      &  \;exponential\;   & \multirow{1}{*}{\;{\bf exponential}\;} & \multirow{1}{*}{\;{\bf memoryless}\;} & \multirow{1}{*}{\;{\bf linear}\;}  \\
\hline                        
Almost-sure \large{\strut} & infinite   & \multirow{2}{*}{{\bf infinite}} & \multirow{2}{*}{{\bf memoryless}} & \multirow{2}{*}{{\bf linear}}  \\
\cline{1-2}
Limit-sure \large{\strut}  & unbounded   &  &  &  \\
\hline
\end{tabular}  
}
\end{center}
\caption{Memory requirement (new results in boldface).\label{tab:memory}}
\end{table}

\paragraph{Decision problems.}
For $f \in \{\fsum_T,\fmax_T\}$ and $\lambda \in \{event, weakly, strongly\}$, 
the \emph{winning region} $\winsure{\lambda}(f)$ is the set of initial distributions such that
$\M$ is sure winning for $\lambda$-synchronizing (we assume that $\M$ is clear from the context). 
We define analogously the sets $\winas{\lambda}(f)$ and $\winlim{\lambda}(f)$.
For a singleton $T = \{q\}$ we have $\fsum_{T} = \fmax_{T}$,
and we simply write $\win{\lambda}{\mu}(q)$ (where $\mu \in \{sure, almost, limit\}$).
It follows from the definitions that 
$\win{strongly}{\mu}(f) \subseteq \win{weakly}{\mu}(f) \subseteq \win{event}{\mu}(f)$
and thus strong and weak synchronization are conservative approximations of eventually synchronization.
It is easy to see that $\winsure{\lambda}(f) \subseteq \winas{\lambda}(f) \subseteq \winlim{\lambda}(f)$,
and for $\lambda = event$ the inclusions are strict~\cite{DMS14}.
In contrast, weak and strong synchronization are more robust
as we show in this paper that the almost-sure and limit-sure winning modes coincide.

\begin{lemma}\label{lem:dif-in-def}
There exists an MDP $\M$ and state $q$ such that
$\winsure{\lambda}(q) \subsetneq \winas{\lambda}(q)$
for $\lambda \in \{weakly, strongly\}$.
\end{lemma}

\begin{proof}
Consider the MDP $\M$ with initial state $q_{\init}$ and action set~$\{a\}$ as shown in 
\figurename~\ref{fig:almost-limit-strongly-differ}. 
On action~$a$ in $q_{\init}$, the successor is $q_{\init}$ or $q$ with probability $\frac{1}{2}$,
and~$q$ is an absorbing state.

We show  that 
$q_{\init} \in \winas{strongly}(q)$ and $q_{\init} \not \in \winsure{strongly}(q)$.
Since $\M$ has only a single action, so it is a Markov chain
with a unique possible strategy~$\alpha$: always playing~$a$.
The outcome under~$\alpha$ is such that the probability
to be in~$q$ after~$k$ steps is $1-\frac{1}{2^k}$ for all~$k$, showing that $\M$
is almost-sure winning for the strongly synchronizing objective in~$\{q\}$ (from~$q_{\init}$).
On the other hand, $q_{\init} \not \in \winsure{strongly}(q_1)$ because
under~$\alpha$, the probability in~$q_{\init}$ remains always positive,
and thus in~$q$ we have $\M^{\alpha}_n(q) < 1$ for all $n \geq 0$,
showing that~$\M$ is not sure winning for the strongly synchronizing 
objective in~$\{q\}$ (from~$q_{\init}$).
The same argument holds for weakly synchronizing objective.
\qed
\end{proof}

\begin{exclude}
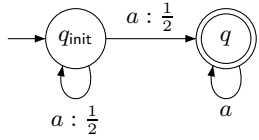
\begin{figure}[t]
\begin{center}
    \begin{picture}(30,15)

\node[Nmarks=i,iangle=180](n0)(0,12){$q_{\init}$}
\node[Nmarks=r](n1)(20,12){$q$}

\drawedge(n0,n1){$a: \frac{1}{2}$}
\drawloop[ELside=l,loopCW=y, loopangle=-90, loopdiam=4](n0){$a:\frac{1}{2}$}
%
\drawloop[ELside=l,loopCW=y, loopangle=-90, loopdiam=4](n1){$a$}


\end{picture}
\end{center}
 \caption{An MDP~$\M$  such that $\winsure{\lambda}(q) \neq \winas{\lambda}(q)$ 
for $\lambda \in \{weakly, strongly\}$.
\label{fig:almost-limit-strongly-differ}}
\end{figure}
\end{exclude}

The \emph{membership problem} is to decide,
given an initial probability distribution $d_0$, whether $d_0 \in \win{\lambda}{\mu}(f)$.
It is sufficient to consider Dirac initial distributions 
(i.e., assuming that MDPs have a single initial state) because the answer to 
the general membership problem for an MDP $\M$ with initial distribution $d_0$ 
can be obtained by solving the membership problem for a copy of $\M$ with a
new initial state from which the successor distribution 
on all actions is $d_0$. 

For eventually synchronizing, the membership problem is PSPACE-complete for all winning modes~\cite{DMS14}. 
In this paper, we show that the complexity of the 
membership problem is PSPACE-complete for weak synchronization, and even 
PTIME-complete for strong synchronization. 
The complexity results are summarized in Table~\ref{tab:complexity},
and we present the memory requirement for winning strategies in Table~\ref{tab:memory}.
\section{Weak Synchronization}

We establish the complexity and memory requirement for weakly synchronizing 
objectives. We show that the membership problem is PSPACE-complete for 
sure and almost-sure winning, that exponential memory is necessary and sufficient
for sure winning while infinite memory is necessary for almost-sure winning,
and we show that limit-sure and almost-sure winning coincide. 


\subsection{Sure weak synchronization}
The PSPACE upper bound of the membership problem for sure weak synchronization is
obtained by the following characterization.


\begin{lemma}\label{lem: suWS}
Let $\M$ be an MDP and $T$ be a target set. For all states $q_{\init}$, 
we have $q_{\init} \in \winsure{weakly}(\fsum_T)$
if and only if there exists a set $S \subseteq T$ such that 
$q_{\init} \in \Pre^{m}(S)$ for some $m \geq 0$ and $S \subseteq \Pre^{n}(S)$
for some $n \geq 1$.
\end{lemma}

\begin{proof}

First, if $q_{\init} \in \winsure{weakly}(\fsum_T)$, then let $\alpha$ be
a sure winning weakly synchronizing strategy.
Then there are infinitely many positions $n$ such that
$\M^{\alpha}_n(T)=1$, and since the state space is finite, 
there is a set $S$ of states 
such that for infinitely many positions~$n$ we have 
$\Supp(\M^{\alpha}_n) = S$ and $\M^{\alpha}_n(T)=1$, and thus $S \subseteq T$.
By the result of~\cite[Lemma~4]{DMS14}, it follows that 
$q_{\init} \in \Pre^{m}(S)$ for some $m \geq 0$, and 
by considering two positions $n_1 < n_2$ 
where $\Supp(\M^{\alpha}_{n_1}) = \Supp(\M^{\alpha}_{n_2}) = S$,
it follows that $S \subseteq \Pre^{n}(S)$ for $n  = n_2 - n_1 \geq 1$. 
%
%

The reverse direction is straightforward by considering
a strategy $\alpha$ that ensures $\M^{\alpha}_m(S)=1$
for some $m \geq 0$, and then ensures that the probability mass 
from all states in $S$ remains in $S$ after every multiple of 
$n$ steps where $n>0$ is such that $S \subseteq \Pre^{n}(S)$,
showing that $\alpha$ is a sure winning weakly synchronizing strategy
in $S$ (and thus in $T$) from $q_{\init}$, thus $q_{\init} \in \winsure{weakly}(\fsum_T)$.
\qed
\end{proof}

The PSPACE upper bound follows from the characterization in 
Lemma~\ref{lem: suWS}. A (N)PSPACE algorithm is to guess the set $S \subseteq T$,
and the numbers $m,n$ (with $m,n \leq 2^{\abs{Q}}$ since the sequence $\Pre^{n}(S)$ of 
predecessors is ultimately periodic), and check that $q_{\init} \in \Pre^{m}(S)$ and
$S \subseteq \Pre^{n}(S)$.
%
%
The PSPACE lower bound  follows from the PSPACE-completeness of the
membership problem for sure eventually synchronization~\cite[Theorem~2]{DMS14}.

\begin{lemma}\label{lem: sure-weakly-psapce-hradness}
The membership problem for $\winsure{weakly}(\fsum_T)$ is PSPACE-hard
even if $T$ is a singleton.
\end{lemma}

\begin{proof}
The proof is by a reduction from the membership problem for
$\winsure{event}(\fsum_T)$ with a singleton~$T$, which is 
PSPACE-complete~\cite[Theorem~2]{DMS14}.  
From an MDP~$\M=\tuple{Q,\Act,\delta}$ with initial state~$q_{\init}$
and target state~$\q$, 
we construct another MDP~$\N=\tuple{Q',\Act',\delta'}$ and a target state~$\p$ 
such that $q_{\init}\in \winsure{event}(\q)$ 
in~$\M$ if and only if $q_{\init}\in \winsure{weakly}(\p)$ in~$\N$.

The MDP~$\N$ is a copy of $\M$ with two new states~$\p$ and $\sink$ 
that are reachable only by a new action~$\sharp$ (see \figurename~\ref{fig:sure-ws-reduction}). 
Formally, $Q' = Q \cup \{\p, \sink\}$ and $\Act' = \Act \cup \{\sharp\}$. 
The transition function $\delta'$ is defined as follows:
$\delta'(q,a) = \delta(q,a)$ for all states $q \in Q$ and $a \in \Act$,
$\delta(q,\sharp)(\sink) = 1$ for all $q \in Q' \setminus \{\q\}$ and $\delta(\q,\sharp)(\p) = 1$.
The state $\sink$ is absorbing and from state~$\p$ all other transitions lead to 
the initial state, i.e. $\delta(\sink,a)(\sink)=1$ and 
$\delta(\p,a)(q_{\init})=1$ for all $a\in \Act$.

\begin{exclude}
\begin{figure}[t]
\begin{center}
    \begin{picture}(115,44)(0,2)

\node[Nmarks=n, Nw=40, Nh=22, dash={0.2 0.5}0](m1)(18,24){}
\node[Nframe=n](label)(6,31){MDP $\M$}

\node[Nmarks=i](n1)(8,22){$q_{\init}$}
\node[Nmarks=r](n1)(32,22){$\q$}
\node[Nmarks=n](n2)(20,22){$q$}
\node[Nframe=n](arrow)(48,24){{\Large $\Rightarrow$}}

\node[Nmarks=n, Nw=40, Nh=22, dash={0.2 0.5}0](nm1)(78,24){}
\node[Nmarks=n, Nw=57, Nh=40, dash={0.4 1}0](m2)(84,20){}
\node[Nframe=n](label)(65,37){MDP $\N$}

\node[Nframe=n](label)(66,31){MDP $\M$}
\node[Nmarks=i](n0)(68,22){$q_{\init}$}
\node[Nmarks=r](n1)(92,22){$\q$}
\node[Nmarks=n](n2)(80,22){$q$}

\node[Nmarks=n](end)(75,5){$\sink$}
\node[Nmarks=n](qq)(106,22){$\p$} 

\drawloop[ELside=l,loopCW=y, loopangle=180, loopdiam=4](end){$\Act'$}

\drawedge[ELpos=40, ELside=l, curvedepth=0](n0,end){$\sharp$}
\drawedge[ELpos=40, ELside=l, curvedepth=0](n2,end){$\sharp$}
\drawedge[ELpos=50, ELside=r, curvedepth=0](n1,qq){$\sharp$}
\drawedge[ELpos=50, ELside=r, curvedepth=-8](qq,n0){$\Act$}
\drawedge[ELpos=50, ELside=l, curvedepth=7](qq,end){$\sharp$}

\end{picture}
\end{center}
 \caption{The reduction sketch to show PSPACE-hardness of the emptiness problem 
for sure weak synchronization in MDPs.}
\label{fig:sure-ws-reduction}
\end{figure}
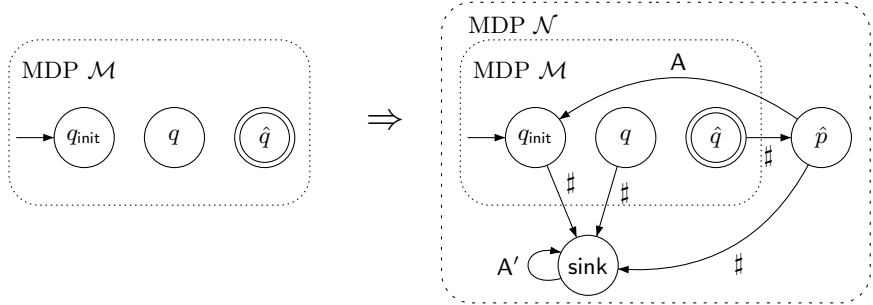
\end{exclude}

We establish the correctness of the reduction as follows.
First, if $q_{\init} \in \winsure{event}(\q)$ in~$\M$, then let $\alpha$
be a sure winning strategy in $\M$ for eventually synchronization in $\{\q\}$.
A sure winning strategy in $\N$ for weak synchronization in $\{\p\}$ is 
to play according to $\alpha$ until the whole probability mass is in $\q$,
then play~$\sharp$ followed by some $a \in \A$ to visit $\p$ and get back to the initial state $q_{\init}$,
and then repeat the same strategy from $q_{\init}$. Hence $q_{\init} \in \winsure{weakly}(\p)$
in $\N$.

Second, if $q_{\init} \in \winsure{weakly}(\p)$ in~$\N$, then consider 
a strategy $\alpha$ such that $\N^{\alpha}_n(\p) = 1$ for some $n \geq 0$.
By construction of $\N$, it follows that $\N^{\alpha}_{n-1}(\q)=1$,
that is all path-outcomes of $\alpha$ of length $n-1$ reach $\q$, and
$\alpha$ plays $\sharp$ in the next step. 
If $\alpha$ never plays $\sharp$ before position $n-1$, then $\alpha$
is a valid strategy in $\M$ up to step $n-1$ and it shows that 
$q_{\init} \in \winsure{event}(\q)$ is sure winning in $\M$ for eventually synchronization in $\{\q\}$.
Otherwise let $m$ be the largest number such that 
there is a finite path-outcome $\rho$ of $\alpha$ of length $m < n-1$ such that 
$\sharp \in \Supp(\alpha(\rho))$. 
Note that the action $\sharp$ can be played by $\alpha$ only in the state $\q$, 
and thus the initial state is reached again after one more step. 
It follows that in some path-outcome $\rho'$ of $\alpha$ 
of length $m+2$, we have $\Last(\rho') = q_{\init}$, and by the choice of $m$,
the action $\sharp$ is not played by $\alpha$ until position $n-1$ where
all the probability mass is in~$\q$. 
Hence the strategy that plays like $\alpha$ from~$\rho'$ in~$\N$ 
is a valid strategy from $q_{\init}$ in $\M$, and is a witness that 
$q_{\init} \in \winsure{event}(\q)$.
\qed
\end{proof}



\begin{exclude}
\begin{figure}[!t]
\begin{center}
\def\fsize{\normalsize}

\begin{picture}(75,72)(0,0)

{\fsize

\node[Nmarks=i, iangle=180](q0)(9,30){$q_{\init}$}
\node[Nmarks=n](q1)(29,51){$q^1_1$}
\node[Nmarks=n](q2)(49,51){$q^1_2$}

\node[Nmarks=n](q3)(29,13){$q^2_1$}
\node[Nmarks=n](q4)(49,23){$q^2_2$}
\node[Nmarks=n](q5)(49,3){$q^2_3$}

\node[Nmarks=r](safe)(69,30){$q_T$}

\node[Nmarks=n, Nw=30, Nh=18, dash={1.5}0, ExtNL=y, NLangle=22, NLdist=1](A1)(39,51){$H_1$}
\node[Nmarks=n, Nw=30, Nh=32, dash={1.5}0, ExtNL=y, NLangle=38, NLdist=1](A2)(39,13){$H_2$}


\drawedge[ELpos=43, ELside=l, ELdist=1, curvedepth=0](q0,q1){$a,b: \frac{1}{2}$}
\drawedge[ELpos=40, ELside=r, ELdist=1, curvedepth=0](q0,q3){$a,b: \frac{1}{2}$}

\drawedge[ELpos=50, ELside=l, ELdist=1, curvedepth=4](q1,q2){$a$}
\drawedge[ELpos=50, ELside=l, ELdist=1, curvedepth=4](q2,q1){$a$}

\drawedge[ELpos=50, ELside=l, ELdist=1, curvedepth=4](q3,q4){$a$}
\drawedge[ELpos=50, ELside=r, ELdist=1.5, curvedepth=4](q4,q5){$a$}
\drawedge[ELpos=40, ELside=l, ELdist=1, curvedepth=4](q5,q3){$a$}

\drawedge[ELpos=50, ELside=l, ELdist=1, curvedepth=0](q2,safe){$b$}
\drawedge[ELpos=50, ELside=r, ELdist=1, curvedepth=0, syo=-3](q5,safe){$b$}



\drawline[AHnb=1, arcradius=1.5](69,34)(69,64)(5,64)(5,34)(6.2,32.8)
\node[Nframe=n](label)(37,66.5){$a,b$}



}
\end{picture}
\caption{The MDP $\M_2$.}\label{fig:exp-mem-weakly}
\end{center}
\end{figure}
\end{exclude}

The proof of Lemma~\ref{lem: suWS} suggests an exponential-memory strategy
for sure weakly synchronization that in $q \in \Pre^{n}(S)$ plays
an action $a$ such that $\post(q,a)\subseteq \Pre^{n-1}(S)$, which  
can be realized with exponential memory since $n \leq 2^{\abs{Q}}$.
It can be shown that exponential memory is necessary in general. 
The argument is very similar to the proof
of exponential memory lower bound for sure eventually synchronization~\cite[Section 4.1]{DMS14}.
For the sake of completeness, we present 
a family of MDPs $\M_n$ ($n \in \nat$) over alphabet $\{a,b\}$ 
that are sure winning for weak synchronization, and where the sure winning strategies 
require exponential memory. The MDP $\M_2$ is shown in \figurename~\ref{fig:exp-mem-weakly}.
The structure of $\M_n$ is an initial uniform probabilistic transition
to $n$ components $H_1, \dots, H_n$ where $H_i$ is a cycle of length $p_i$
the $i$-th prime number. On action $a$, the next state in the cycle is reached,
and on action $b$ the target state $q_T$ is reached, only from the last
state in the cycles. From other states, 
the action $b$ leads to an absorbing  sink state (transitions not depicted).
A sure winning strategy from $q_{\init}$ for weak synchronization in $\{q_T\}$ is to 
play $a$ in the first $p^{\#}_n = \prod_{i=1}^{n} p_i$ steps, and then play $bb$
to reach $q_{\init}$ again, through  $q_T$.
This requires memory of size $p^{\#}_n > 2^n$ while the size of $\M_n$
is in $O(n^2 \log n)$~\cite{BS96}.
It can be proved that all winning strategies for weak synchronization
need to be, from $q_{\init}$, sure eventually synchronizing in $\{q_T\}$
(consider the last occurrence of $q_{\init}$ along a play before all
the probability mass is in $q_T$) and this requires memory of size at least
$p^{\#}_n$ by standard pumping arguments as in~\cite{DMS14}.

\begin{theorem} \label{theo:weakly-sure}
For sure weak synchronization in MDPs:
\begin{enumerate}
\item  (Complexity). The membership problem is PSPACE-complete.
\item  (Memory). Exponential memory is necessary and sufficient for both pure and
randomized strategies, and pure strategies are sufficient.
\end{enumerate}
\end{theorem}

\subsection{Almost-sure weak synchronization}\label{sec:almost-weakly}


We present a characterization of almost-sure weak synchronization that 
gives a PSPACE upper bound for the membership problem. Our characterization
uses the limit-sure eventually synchronizing objectives \emph{with exact support}~\cite{DMS14}.  
This objective requires that the probability mass tends to $1$ in a target
set $T$, and moreover that after the same number of
steps the support of the probability distribution is contained in a
given set $U$.  Formally, given an MDP $\M$, let
$\winlim{event}(\fsum_T,U)$ for $T \subseteq U$ be the set of all
initial distributions such that for all $\epsilon>0$ there exists a
strategy~$\alpha$ and $n \in \nat$ such that $\M^{\alpha}_n(T) \geq
1-\epsilon$ and $\M^{\alpha}_n(U)=1$.

We show that 
an MDP is almost-sure weakly synchronizing in target $T$ if (and only if),
for some set $U$, there is a sure eventually synchronizing strategy in target $U$,
and from the probability distributions with support $U$ there is a limit-sure
winning strategy for eventually synchronizing in $\Pre(T)$ with support in $\Pre(U)$.
This ensures that from the initial state we can have the whole probability mass in~$U$, and from~$U$ 
have probability $1-\epsilon$ in $\Pre(T)$ (and in $T$ in the next step), 
while the whole probability mass is back in $\Pre(U)$ (and in $U$ in the next step), 
allowing to repeat the strategy for $\epsilon \to 0$, thus ensuring infinitely
often probability at least $1-\epsilon$ in $T$ (for all $\epsilon > 0$).

\begin{lemma}\label{lem: almost-weak-reduce-limit-event}
Let $\M$ be an MDP and $T$ be a target set. For all states $q_{\init}$, 
we have $q_{\init}\in \winas{weakly}(\fsum_T)$
if and only if there exists a set~$U$ such that
\begin{itemize}
\item $q_{\init} \in \winsure{event}(\fsum_U)$, and \smallskip
\item $d_U \in \winlim{event}(\fsum_{\Pre(T)},\Pre(U))$
where $d_U$ is the uniform distribution over~$U$.
\end{itemize}
\end{lemma}

\begin{proof}
First, if $q_{\init} \in \winas{weakly}(\fsum_T)$, then 
there exists a strategy $\alpha$ such that for all $i \geq 0$ 
there exists $n_i \in \nat$ such that $\M^{\alpha}_{n_i}(T) \geq 1-2^{-i}$,
and moreover $n_{i+1} > n_i$ for all $i \geq 0$.
Let $s_i = \Supp(\M^{\alpha}_{n_i})$ be the support of~$\M^{\alpha}_{n_i}$.
Since the state space is finite, there is a set~$U$ that 
occurs infinitely often in the sequence~$s_0 s_1 \dots$,
thus for all $k>0$ there exists $m_k \in \nat$ such that 
$\M^{\alpha}_{m_k}(T) \geq 1-2^{-k}$ and 
$\M^{\alpha}_{m_k}(U) = 1$.
It follows that $\alpha$ is sure eventually synchronizing in $U$ 
from $q_{\init}$, i.e. $q_{\init} \in \winsure{event}(\fsum_U)$.
Moreover, we can assume that $m_{k+1} > m_k$ for all $k>0$
and thus $\M$ is also limit-sure eventually synchronizing in $\Pre(T)$
with exact support in $\Pre(U)$ from the initial distribution $d_1 = \M^{\alpha}_{m_1}$.
Since $\Supp(d_1) = U = \Supp(d_U)$ and since 
only the support of the initial probability distributions is relevant for the limit-sure
eventually synchronizing objective~\cite[Corollary 1]{DMS14},
it follows that $d_U \in \winlim{event}(\fsum_{\Pre(T)},\Pre(U))$.

To establish the converse, note that 
since $d_U \in \winlim{event}(\fsum_{\Pre(T)},\Pre(U))$, it follows
from \cite[Corollary 1]{DMS14} that
from all initial
distributions with support in~$U$, for all $\epsilon > 0$ there exists
a strategy $\alpha_{\epsilon}$ and a position $n_{\epsilon}$ such that
$\M^{\alpha_{\epsilon}}_{n_{\epsilon}}(T) \geq 1-\epsilon$ and
$\M^{\alpha_{\epsilon}}_{n_{\epsilon}}(U) = 1$.  We construct an
almost-sure weakly synchronizing strategy $\alpha$ as
follows.  Since $q_{\init} \in \winsure{event}(\fsum_U)$, play according to
a sure eventually synchronizing strategy from $q_{\init}$ until all the
probability mass is in~$U$.  Then for $i=1,2, \dots$ and $\epsilon_i =
2^{-i}$, repeat the following procedure: given the current probability distribution, 
select the corresponding strategy $\alpha_{\epsilon_i}$
and play according to $\alpha_{\epsilon_i}$ for $n_{\epsilon_i}$
steps, ensuring probability mass at least $1-2^{-i}$ in $\Pre(T)$ and 
support of the probability mass in $\Pre(U)$. Then from states in $\Pre(T)$,
play an action to ensure reaching $T$ in the next step, and from states
in $\Pre(U)$ ensure reaching $U$. Continue playing according to 
$\alpha_{\epsilon_{i+1}}$ for $n_{\epsilon_{i+1}}$ steps, etc.  
Since $n_{\epsilon_i} + 1 > 0$ for all $i \geq 0$, this strategy ensures that $\limsup_{n \to \infty}
\M^{\alpha}_n(T) = 1$ from $q_{\init}$, hence $q_{\init} \in \winas{weak}(\fsum_T)$.

\qed
\end{proof}

Since the membership problems for sure eventually synchronizing
and for limit-sure eventually synchronizing with exact support are PSPACE-complete 
(\cite[Theorem~2 and~4]{DMS14}),
%
%
the membership problem for almost-sure weak synchronization
is in PSPACE by  
guessing the set $U$, and checking that $q_{\init} \in \winsure{event}(\fsum_U)$, and that 
$d_U \in \winlim{event}(\fsum_{\Pre(T)},\Pre(U))$.
We establish a matching PSPACE lower bound.

\begin{lemma}\label{lem:almost-weakly-pspace-hard}
The membership problem for $\winas{weakly}(\fsum_T)$ is PSPACE-hard
even if $T$ is a singleton.
\end{lemma}

\begin{proof}
%
The problem of deciding, given an MDP $\M$ and a singleton $T$, 
whether $\Pre^n_{\M}(T) \neq \emptyset$ for all $n \geq 0$ is 
PSPACE-complete~\cite[Lemma~3]{DMS14}. We present a reduction 
of this problem to the membership problem for almost-sure weak synchronization,
very similar to the proof of PSPACE-hardness for limit-sure eventually synchronizing~\cite[Lemma~11]{DMS14}.

The reduction is as follows (see also \figurename~\ref{fig:lim-sure-reduction}). 
Given an MDP $\M=\tuple{Q, \Act,\delta}$ and a singleton $T \subseteq Q$,
we construct an MDP $\N =\tuple{Q', \Act',\delta'}$ with state space $Q' = Q \uplus \{q_{\init}\}$ such that 
$\Pre^n_{\M}(T) \neq \emptyset$ for all $n \geq 0$ if and only if 
$q_{\init}$ is almost-sure weakly synchronizing in $T$.
The MDP $\N$ is essentially a copy of $\M$ with alphabet $\Act \uplus \{\sharp\}$ and 
the transition function on action $\sharp$ is the uniform 
distribution on $Q$ from $q_{\init}$, and the Dirac distribution on 
$q_{\init}$ from the other states $q \in Q$. There are self-loops on $q_{\init}$ 
for all other actions $a \in \Act$. Formally, the
transition function $\delta'$ is defined as follows, for all $q \in Q$:

\begin{itemize} 
\item $\delta'(q,a) = \delta(q,a)$ for all $a \in \Act$ (copy of $\M$),
      and $\delta'(q,\sharp)(q_{\init}) = 1$;
\item $\delta'(q_{\init},a)(q_{\init}) = 1$ for all $a \in \Act$,
      and $\delta'(q_{\init}, \sharp)(q) = \frac{1}{\abs{Q}}$.
\end{itemize}

We establish the correctness of the reduction as follows.
For the first direction, assume that 
$\Pre^n_{\M}(T) \neq \emptyset$ for all~$n \geq 0$. 
It follows that 
there exist numbers $k_0,r \leq 2^{\abs{Q}}$ 
such that $\Pre^{k_0 + r}_{\M}(T) = \Pre^{k_0}_{\M}(T) \neq \emptyset$.

By Lemma~\ref{lem: almost-weak-reduce-limit-event} with $U = Q$, 
we need to show that $(i)$ $q_{\init} \in \winsure{event}(\fsum_Q)$, and 
$(ii)$ $d_Q \in \winlim{event}(\fsum_{\Pre(T)},\Pre(Q))$
where $d_Q$ is the uniform distribution over~$Q$.
To show $(i)$, we can play $\sharp$ from $q_{\init}$ to get
the probability mass synchronized in $Q$. To show $(ii)$, since playing
$\sharp$ from $d_Q$ ensures to reach $q_{\init}$, it suffices to prove 
that $q_{\init} \in \winlim{event}(\fsum_{T},Q)$,
and it is sufficient to prove this in $\M$ since $\N$ embeds a copy of $\M$
(note that the requirement that the exact support is in $Q$ becomes trivial then).
Using~\cite[Lemma~8]{DMS14}
with $k = k_0$ and $R = \Pre^{k_0}_{\M}(T)$ 
(and $U = Z = Q$ is the trivial support), it is sufficient to
prove that $q_{\init} \in \winlim{event}(\fsum_{R})$ to get $q_{\init} \in \winlim{event}(\fsum_{T})$.
We show the stronger statement that $q_{\init}$ is actually almost-sure 
eventually synchronizing in $R$ with the pure strategy $\alpha$ defined as follows,
for all play prefixes $\rho$ (let $m = \abs{\rho} \!\!\mod r$):

\begin{itemize} 
\item if $\Last(\rho) = q_{\init}$, then $\alpha(\rho) = \sharp$;
\item if $\Last(\rho) = q \in Q$, then 
	\begin{itemize} 
	\item if $q \in \Pre^{r-m}_{\M}(R)$ for $0 \leq m < r$, 
	then $\alpha(\rho) = a$ such that $\post(q,a) \subseteq \Pre^{r-m-1}_{\M}(R)$;
	\item otherwise, $\alpha(\rho) = \sharp$.
	\end{itemize}
\end{itemize}

\begin{exclude}
\begin{figure}[t]
\begin{center}
\begin{picture}(115,46)(0,2)

\node[Nmarks=n, Nw=40, Nh=22, dash={0.2 0.5}0](m1)(20,20){}
\node[Nframe=n](label)(8,13){MDP $\M$}
\drawpolygon[dash={0.8 0.5}0](38,30)(23,30)(23,10)(38,10)
\node[Nframe=n](label)(30,13){$T\subseteq Q$}
\node[Nmarks=r](n1)(30,20){$q_2$}

\node[Nframe=n](label)(19,20){$\dots$}
\node[Nmarks=n](n2)(10,20){$q_1$}
\node[Nframe=n](arrow)(45,20){{\Large $\Rightarrow$}}

\node[Nmarks=n, Nw=40, Nh=22, dash={0.2 0.5}0](nm1)(80,20){}
\node[Nmarks=n, Nw=50, Nh=48, dash={0.4 1}0](m2)(80,32){}
\node[Nframe=n](label)(65,52){MDP $\N$}

\node[Nframe=n](label)(68,13){MDP $\M$}
\drawpolygon[dash={0.8 0.5}0](98,30)(83,30)(83,10)(98,10)
\node[Nframe=n](label)(90,13){$T\subseteq Q$}
\node[Nmarks=r](nn1)(90,20){$q_2$}
\node[Nframe=n](label)(79,20){$\dots$}
\node[Nmarks=n](nn2)(70,20){$q_1$}

\node[Nmarks=i](qq)(80,43){$q_{\init}$} 

\drawloop[ELside=l,loopCW=y, loopangle=90, loopdiam=4](qq){$\Act$}

\drawedge[ELpos=50, ELside=l](qq,nn1){$\sharp$}
\node[Nframe=n](label)(80.2,36){$\cdots$}
\drawedge[ELpos=50, ELside=r](qq,nn2){$\sharp$}

\drawedge[ELpos=50, ELside=r, curvedepth=-8](nn1,qq){$\sharp$}

\drawedge[ELpos=50, ELside=l, curvedepth=+8](nn2,qq){$\sharp$}

\end{picture}
\caption{Sketch of reduction to show PSPACE-hardness of  the membership problem
for almost-sure weak synchronization.}\label{fig:lim-sure-reduction}
\end{center}
\end{figure}
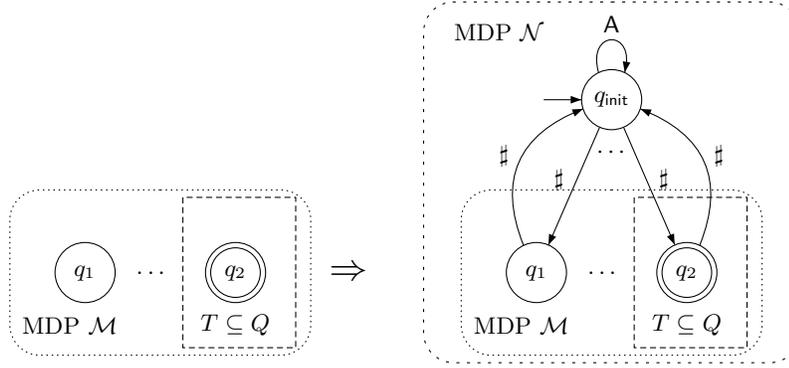
\end{exclude}

Note that if $q \in R$, then $q \in \Pre^{r-m}_{\M}(R)$ for $m=0$ since $\Pre^{r}_{\M}(R) = R$.
The strategy $\alpha$ ensures that the probability mass that is not (yet) in 
the sequence of predecessors $\Pre^n_{\M}(R)$ goes to $q_{\init}$, where by playing $\sharp$
at least a fraction $\frac{1}{\abs{Q}}$ of it would reach the sequence of predecessors
(at a synchronized position). It follows that after $2i$ steps, the probability
mass in $q_{\init}$ is at most $(1 - \frac{1}{\abs{Q}})^i$ and the probability mass
in the sequence of predecessors is at least $1 - (1 - \frac{1}{\abs{Q}})^i$. 
For $i \to \infty$,
the probability in the sequence of predecessors tends to $1$ and since 
$\Pre^n_{\M}(R) = R$ for all positions $n$ that are a multiple of $r$,
we get $\sup_{n}\M^{\alpha}_{n}(R) = 1$ and $q_{\init} \in \winas{event}(\fsum_{R})$.

For the converse direction, assume that $q_{\init}$ is almost-sure weakly synchronizing in~$T$,
then $q_{\init}$ is also limit-sure eventually synchronizing in~$T$.
By~\cite[Lemma~8]{DMS14},
%
%
either $(1)$ $q_{\init}$ is limit-sure eventually synchronizing in $\Pre^n_{\N}(T)$
for all $n \geq 0$, and then it follows that $\Pre^n_{\N}(T) \neq \emptyset$ 
for all $n \geq 0$, or $(2)$ $q_{\init}$ is sure eventually synchronizing in $T$,
and then since only the action $\sharp$ leaves the state $q_{\init}$ (and $\post(q_{\init},\sharp) = Q$), 
and since  $q_{\init} \in \winsure{event}(\fsum_T)$ 
if and only if there exists $k \geq 0$ such that $q_{\init}\in \Pre_{\N}^{k}(T)$~\cite[Lemma~4]{DMS14}, 
we have $Q \subseteq \Pre^{k-1}_{\N}(T)$.
Moreover, since 
$Q \subseteq \Pre_{\N}(Q)$ and $\Pre_{\N}(\cdot)$ is a monotone operator, 
it follows that $Q \subseteq \Pre^n_{\N}(T)$ for all $n \geq k-1$ and thus 
$\Pre^n_{\N}(T) \neq \emptyset$ for all $n \geq 0$.
We conclude the proof by noting that $\Pre^n_{\M}(T) = \Pre^n_{\N}(T) \cap Q$
and therefore $\Pre^n_{\M}(T) \neq \emptyset$ for all $n \geq 0$.
\qed 


\end{proof}

Simple examples show that winning strategies require infinite memory
for almost-sure weak synchronization.


\begin{exclude}
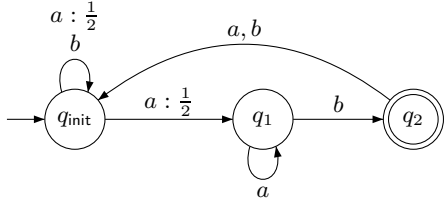
\begin{figure}[t]
\begin{center}
\begin{picture}(60,26)

\node[Nmarks=i,iangle=180](n0)(10,10){$q_{\init}$}
\node[Nmarks=n](n1)(35,10){$q_1$}
\node[Nmarks=r](n2)(55,10){$q_2$}

\drawedge[ELdist=.5](n0,n1){$a: \frac{1}{2}$}
\drawloop[ELside=l,ELdist=0, loopCW=y, loopangle=90, loopdiam=4](n0){$\begin{array}{c}a:\frac{1}{2}\\ b\end{array}$}

\drawedge(n1,n2){$b$}
\drawloop[ELside=r,loopCW=n, loopangle=-90, loopdiam=4](n1){$a$}

\drawedge[ELpos=50, ELdist=.5, ELside=r, curvedepth=-10](n2,n0){$a,b$}

\end{picture}
\caption{An MDP where infinite memory is necessary for
almost-sure weakly synchronizing strategies.}\label{fig:inf-mem}
\end{center}
\end{figure}
\end{exclude}



\begin{theorem}\label{theo:weakly-almost}
For almost-sure weak synchronization in MDPs:

\begin{enumerate}
\item (Complexity). The membership problem is PSPACE-complete.

\item (Memory). Infinite memory is necessary in general for both pure and 
randomized strategies, and pure strategies are sufficient.
\end{enumerate}
\end{theorem}

\begin{proof}
The result on memory requirement is established by following example.
The example and argument are analogous to the proof that infinite memory 
is necessary for almost-sure eventually synchronizing~\cite[Section~4.2]{DMS14}.
Consider the MDP $\M$  shown  in~\figurename~\ref{fig:inf-mem}
with three states $q_{\init},q_1,q_2$ and two actions $a,b$. 
The only probabilistic transition is in $q_{\init}$ on action $a$ that has 
successors $q_{\init}$ and $q_1$ with probability $\frac{1}{2}$. 
The other transitions are deterministic. Let $q_{\init}$ be the initial state. 
We construct a strategy that is almost-sure
weakly synchronizing in $\{q_2\}$, showing that $q_{\init} \in \winas{weakly}(q_2)$.  
First, observe that for all $\epsilon > 0$  we
can have probability at least $1 - \epsilon$ in $q_2$ after finitely
many steps from $q_{\init}$: playing 
$n$ times $a$ and then
$b$ leads to probability $1- \frac{1}{2^n}$ in $q_2$.
Note that after that, the current probability distribution has support
$\{q_{\init},q_2\}$ and that from such a distribution, we can as well
ensure probability at least $1 - \epsilon$ in $q_2$.
 Thus for a fixed~$\epsilon$, the MDP
 is $(1-\epsilon)$-synchronizing in $\{q_2\}$ (after finitely many steps),
 and by taking a smaller value of $\epsilon$, we can continue to play a strategy 
 to have probability at least $1 - \epsilon$ in $q_2$, and repeat this for $\epsilon \to 0$. 
 This strategy ensures almost-sure weak synchronization in $\{q_2\}$.  Below,
 we show that infinite memory is necessary for almost-sure
 winning in this MDP. 

Assume towards contradiction that there exists a finite-memory
strategy $\alpha$ that is almost-sure weakly synchronizing in
$\{q_2\}$. Consider the Markov chain $\M(\alpha)$ (the product of the MDP
$\M$ with the finite-state transducer defining $\alpha$). A state
$(q,m)$ in $\M(\alpha)$ is called a \emph{$q$-state}. Since $\alpha$
is almost-sure weakly synchronizing 
in $\{q_2\}$, there is a $q_2$-state in the recurrent states
of $\M(\alpha)$.  Since on all actions $q_{\init}$ is a successor of $q_2$,
and $q_{\init}$ is a successor of itself, it follows that there is a
recurrent $q_{\init}$-state in $\M(\alpha)$, and that all periodic classes
of recurrent states in $\M(\alpha)$ contain a $q_{\init}$-state. Hence, in
each stationary distribution there is a $q_{\init}$-state with a positive
probability, and therefore the probability mass in $q_{\init}$ is bounded
away from zero. It follows that the probability mass in $q_2$ is
bounded away from $1$ thus $\alpha$ is not almost-sure weakly
synchronizing in $\{q_2\}$, a contradiction.  \qed
\end{proof}

\subsection{Limit-sure weak synchronization}

We show that the winning regions for almost-sure
and limit-sure weak synchronization coincide. The result
is not intuitively obvious (recall that it does not
hold for eventually synchronizing) and requires a careful
analysis of the structure of limit-sure winning strategies
to show that they always induce the existence of an almost-sure 
winning strategy. The construction of an almost-sure
winning strategy from a family of limit-sure winning strategies
is illustrated in the following example.

Consider the MDP $\M$ in~\figurename~\ref{fig:weak-limit}
with initial state $q_{\init}$ and target set $T = \{q_4\}$.
Note that there is a relevant strategic choice only in $q_3$,
and that $q_{\init}$ is limit-sure winning for eventually synchronization
in $\{q_4\}$ since we can inject a probability mass arbitrarily
close to $1$ in $q_3$ (by always playing $a$ in $q_3$),
and then switching to playing $b$ in $q_3$ gets probability $1 - \epsilon$ 
in $T$ (for arbitrarily small $\epsilon$). Moreover, the same 
holds from state~$q_4$. These two facts are  sufficient to show that $q_{\init}$
is limit-sure winning for weak synchronization in $\{q_4\}$: given
$\epsilon > 0$, play from $q_{\init}$ a strategy to ensure probability
at least $p_1 = 1-\frac{\epsilon}{2}$ in $q_4$ (in finitely many steps),
and then play according to a strategy that ensures from $q_4$ 
probability $p_2  = p_1 - \frac{\epsilon}{4}$ in $q_4$ (in finitely many, and at least one step), 
and repeat this process using strategies that ensure, if the probability mass in $q_4$ 
is at least $p_i$, that the probability in $q_4$ is at least
$p_{i+1}  = p_i - \frac{\epsilon}{2^{i+1}}$ (in at least one step).
It follows that $p_i = 1 - \frac{\epsilon}{2} - \frac{\epsilon}{4} - \dots - \frac{\epsilon}{2^i} > 1 - \epsilon$ for all $i \geq 1$, and thus 
$\limsup_{i \to \infty} p_i \geq 1 - \epsilon$ showing that
$q_{\init}$ is limit-sure weakly synchronizing in target~$\{q_4\}$.

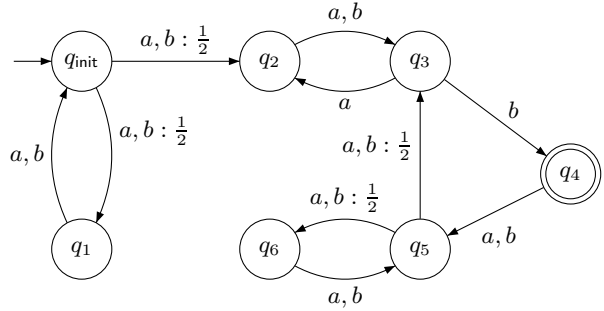
\begin{figure}[t]

\begin{center}
    \begin{picture}(80,40)

\node[Nmarks=i](n0)(10,30){$q_{\init}$}
\node[Nmarks=n](n1)(10,5){$q_1$}
\node[Nmarks=n](n2)(35,30){$q_2$}
\node[Nmarks=n](n3)(55,30){$q_3$}
\node[Nmarks=r](n4)(75,15){$q_4$}
\node[Nmarks=n](n5)(55,5){$q_5$}
\node[Nmarks=n](n6)(35,5){$q_6$}

\drawedge[ELpos=40, ELside=l, curvedepth=4](n0,n1){$a,b:\frac{1}{2}$}
\drawedge[ELpos=50, ELside=l](n0,n2){$a,b:\frac{1}{2}$}
\drawedge[ELpos=50, ELside=l, curvedepth=4](n1,n0){$a,b$}

\drawedge[ELpos=50, ELside=l, curvedepth=4](n2,n3){$a,b$}
\drawedge[ELpos=50, ELside=l, curvedepth=4](n3,n2){$a$}

\drawedge[ELpos=55](n3,n4){$b$}
\drawedge[ELpos=55](n4,n5){$a,b$}
\drawedge[ELpos=55](n5,n3){$a,b:\frac{1}{2}$}

\drawedge[ELpos=50, ELside=r, curvedepth=-4](n5,n6){$a,b:\frac{1}{2}$}
\drawedge[ELpos=50, ELside=r, curvedepth=-4](n6,n5){$a,b$}


\end{picture}
\end{center}
 \caption{An example to show $q_{\init}\in \winlim{weakly}(q_4)$ implies $q_{\init}\in \winas{weakly}(q_4)$.\label{fig:weak-limit}}

\end{figure}

It follows from the result that we establish in this section (Theorem~\ref{theo:weakly-ls-is-as})
that $q_{\init}$ is actually almost-sure weakly synchronizing in target $\{q_4\}$.
To see this, consider the sequence $\Pre^i(T)$ for $i \geq 0$: 
$\{q_4\}, \{q_3\}, \{q_2\}, \{q_3\}, \dots$ is ultimately periodic with period $r=2$
and $R = \{q_3\} = \Pre(T)$ is such that $R = \Pre^2(R)$. The period corresponds
to the loop $q_2 q_3$ in the MDP. It turns out that \emph{limit-sure} eventually
synchronizing in $T$ implies \emph{almost-sure} eventually
synchronizing in $R$ (by the proof of~\cite[Lemma~9]{DMS14}), thus from $q_{\init}$ a 
\emph{single} strategy ensures that the probability mass in $R$ is $1$, either
in the limit or after finitely many steps.
Note that in both cases since $R = \Pre^r(R)$ this even implies almost-sure weakly
synchronizing in $R$. The same holds from state $q_4$.

Moreover, note that all distributions produced by an almost-sure weakly synchronizing 
strategy are themselves almost-sure weakly synchronizing.
An almost-sure winning strategy for 
weak synchronization in $\{q_4\}$ consists in playing from $q_{\init}$ an 
\emph{almost-sure} eventually synchronizing strategy in target $R = \{q_3\}$, 
and considering a decreasing sequence $\epsilon_i$ such that 
$\lim_{i \to \infty} \epsilon_i = 0$, when the probability mass in $R$
is at least $1-\epsilon_i$, inject it in $T = \{q_4\}$. Then the remaining
probability mass defines a distribution (with support $\{q_1,q_2\}$ in the example)
that is still almost-sure eventually
synchronizing in $R$, as well as the states in $T$. 
Note that in the example, (almost all) the probability mass in $T = \{q_4\}$ can move
to $q_3$ in an even number of steps, while from $\{q_1,q_2\}$ an odd number
of steps is required, resulting in a \emph{shift} of the probability mass. 
However, by repeating the strategy two times from $q_4$ 
(injecting large probability mass in $q_3$, moving to $q_4$, and injecting
in $q_3$ again), we can make up for 
the shift and reach $q_3$ from $q_4$ in an even number of steps, thus in 
synchronization with the probability mass from~$\{q_1,q_2\}$. 
This idea is formalized in the rest of this section, and we prove that 
we can always make up for the shifts,
which requires a carefully analysis of the allowed
amounts of shifting.  

The result is easier to prove when the target $T$ 
is a singleton, as in the example. For an arbitrary target set $T$, we need
to get rid of the states in $T$ that do not contribute a significant (i.e., bounded away from $0$)
probability mass in the limit, that we call the `vanishing' states.
We show that they  can be removed from $T$ without changing the winning 
region for limit-sure winning. When the target set has no vanishing state, we can 
construct an almost-sure winning strategy as in the case of a singleton target set.

Given an MDP $\M$ with initial state $q_{\init} \in \winlim{weakly}(\fsum_T)$ that is 
limit-sure winning for the weakly synchronizing objective in target set $T$,
let~$(\alpha_i)_{i\in\nat}$ be a family of limit-sure winning strategies
such that $\limsup_{n \to \infty}\M^{\alpha_i}_n(T)\geq 1-\epsilon_i$
where $\lim_{i \to \infty} \epsilon_i = 0$. Hence by definition of $\limsup$,
for all $i \geq 0$ there exists a strictly increasing sequence
$k_{i,0} < k_{i,1} < \cdots$ of positions such that $\M^{\alpha_i}_{k_{i,j}}(T) \geq 1-2 \epsilon_i$ 
for all $j \geq 0$.
A state $q \in T$ is \emph{vanishing} if $\liminf_{i\to \infty} \liminf_{j\to \infty} \M^{\alpha_i}_{k_{i,j}}(q)=0$
for some family of limit-sure weakly synchronizing strategies~$(\alpha_i)_{i\in\nat}$.
Intuitively, the contribution of a vanishing state $q$ to the probability in $T$
tends to $0$ and therefore $\M$ is also limit-sure winning for the weakly 
synchronizing objective in target set $T \setminus \{q\}$.

\begin{lemma}\label{lem:weak-without-vanishing}
If an MDP $\M$ is limit-sure weakly synchronizing in target set~$T$, 
then there exists a set $T' \subseteq T$ such that 
$\M$ is limit-sure weakly synchronizing in~$T'$
without vanishing states.
\end{lemma}

\begin{proof}
If there is no vanishing state for~$(\alpha_i)_{i\in\nat}$, then take $T' = T$
and the proof is complete. 
Otherwise, let $(\alpha_i)_{i\in\nat}$ be a 
family of limit-sure winning strategies such that $\limsup_{n \to \infty}\M^{\alpha_i}_n(T)\geq 1-\epsilon_i$
where $\lim_{i \to \infty} \epsilon_i = 0$ and let $q$ be a vanishing state 
for $(\alpha_i)_{i\in\nat}$.
We show that $(\alpha_i)_{i\in\nat}$ is limit-sure weakly synchronizing in~$T \setminus \{q\}$.
For every $i \geq 0$ let $k_{i,0} < k_{i,1} < \cdots$ be a strictly increasing sequence 
such that $(a)$ $\M^{\alpha_i}_{k_{i,j}}(T) \geq 1-2\epsilon_i$ for all $i,j \geq 0$, and
$(b)$ $\liminf_{i\to \infty} \liminf_{j\to \infty} \M^{\alpha_i}_{k_{i,j}}(q)=0$.

It follows from $(b)$ that for all $\epsilon > 0$ and all~$x > 0$ there exists~$i>x$ such that
for all~$y > 0$ there exists~$j>y$ such that~$\M^{\alpha_i}_{k_{i,j}}(q)<\epsilon$,
and thus 
$$\M^{\alpha_i}_{k_{i,j}}(T \setminus \{q\}) \geq 1-2\epsilon_i - \epsilon$$
by $(a)$. Since this holds for infinitely many $i$'s, we can choose $i$ such that
$\epsilon_i < \epsilon$ and we have 
$$\limsup_{j \to \infty} \M^{\alpha_i}_{k_{i,j}}(T \setminus \{q\}) \geq 1- 3 \epsilon$$
and thus 
$$\limsup_{n \to \infty} \M^{\alpha_i}_{n}(T \setminus \{q\}) \geq 1- 3 \epsilon$$
since the sequence $(k_{i,j})_{j\in\nat}$ is strictly increasing.
This shows that $(\alpha_i)_{i\in\nat}$ is limit-sure weakly synchronizing in~$T \setminus \{q\}$.

By repeating this argument as long as there is a vanishing state (thus at most $\abs{T}-1$ times), 
we can construct the desired set $T' \subseteq T$ without vanishing state.
\qed 
\end{proof}

For a limit-sure weakly synchronizing MDP in target set $T$ (without vanishing 
states), we show that from a probability distribution with support $T$,
a probability mass arbitrarily close to $1$ can be injected synchronously
back in $T$ (in at least one step), that is $d_T \in \winlim{event}(\fsum_{\Pre(T)})$.
The same holds from the initial state $q_{\init}$ of the MDP. This property
is the key to construct an almost-sure weakly synchronizing strategy.

\begin{lemma}\label{lem:weak-event-uniform}
If an MDP $\M$ with initial state $q_{\init}$ is limit-sure weakly synchronizing in a target set~$T$ without 
vanishing states, then $q_{\init} \in \winlim{event}(\fsum_{\Pre(T)})$ and
$d_T \in \winlim{event}(\fsum_{\Pre(T)})$
where $d_T$ is the uniform distribution over~$T$.
\end{lemma}

\begin{proof}
Since $q_{\init} \in \winlim{weakly}(\fsum_T)$ and 
$\winlim{weakly}(\fsum_T) \subseteq \winlim{event}(\fsum_T)$,
we have $q_{\init} \in \winlim{event}(\fsum_T)$ and thus it suffices
to prove that $d_T \in \winlim{event}(\fsum_{\Pre(T)})$. This is because
then from $q_{\init}$, probability arbitrarily close to $1$ can be injected
in $\Pre(T)$ through a distribution with support in $T$ (since 
by~\cite[Corollary~1]{DMS14} only the support of the 
initial probability distribution is important for limit-sure 
eventually synchronizing).

Let~$(\alpha_i)_{i\in\nat}$ be a family of limit-sure winning strategies
such that $\limsup_{n \to \infty}\M^{\alpha_i}_n(T)\geq 1-\epsilon_i$
where $\lim_{i \to \infty} \epsilon_i = 0$, and such that
there is no vanishing state. 
For every $i \geq 0$ let $k_{i,0} < k_{i,1} < \cdots$ be a strictly increasing sequence 
such that $\M^{\alpha_i}_{k_{i,j}}(T) \geq 1 - 2\epsilon_i$ for all $i,j \geq 0$, and
let $B = \min_{q \in T} \liminf_{i\to \infty}\liminf_{j\to \infty} \M^{\alpha_i}_{k_{i,j}}(q)$.
Note that $B > 0$ since there is no vanishing state.
It follows that there  exists~$x > 0$ such that for all~$i > x$ 
there exists~$y_i > 0$ such that for all~$j > y_i$ and all $q \in T$
we have~$\M^{\alpha_i}_{k_{i,j}}(q) \geq \frac{B}{2}$.

Given~$\nu > 0$, let $i>x$ such that $\epsilon_i <  \frac{\nu B}{4}$,
and for $j > y_i$, consider the positions $n_1 = k_{i,j}$ and $n_2 = k_{i,j+1}$.
We have $n_1 < n_2$ and $\M^{\alpha_i}_{n_1}(T) \geq 1 - 2\epsilon_i$
and $\M^{\alpha_i}_{n_2}(T) \geq 1 - 2\epsilon_i$, and 
$\M^{\alpha_i}_{n_1}(q) \geq \frac{B}{2}$ for all $q \in T$. 
Consider the strategy $\beta$ that plays 
like $\alpha_i$ plays from position $n_1$ and thus transforms the
distribution $\M^{\alpha_i}_{n_1}$ into $\M^{\alpha_i}_{n_2}$. 
For all states $q \in T$, from the Dirac distribution on $q$ under strategy $\beta$, the probability 
to reach $Q \setminus T$ in $n_2 - n_1$ steps is thus at most 
$\frac{\M^{\alpha_i}_{n_2}(Q \setminus T)}{\M^{\alpha_i}_{n_1}(q)} \leq \frac{2\epsilon_i}{B/2} < \nu$.

Therefore, from an arbitrary probability distribution with support $T$
we have $\M^{\beta}_{n_2-n_1}(T) > 1 - \nu$, showing that $d_T$ is
limit-sure eventually synchronizing in $T$ and thus in $\Pre(T)$ 
since $n_2 - n_1 > 0$ (it is easy to show that if the mass of probability in
$T$ is at least $1-\nu$, then the mass of probability in $\Pre(T)$ one step
before is at least $1-\frac{\nu}{\eta}$ where $\eta$ is the smallest positive
probability in $\M$).
\qed
\end{proof}

To show that limit-sure and almost-sure winning coincide for weakly synchronizing
objectives, from a family of limit-sure winning strategies we construct 
an almost-sure winning strategy that uses the eventually synchronizing 
strategies of Lemma~\ref{lem:weak-event-uniform}. The construction consists
in using successively strategies that ensure probability mass $1-\epsilon_i$
in the target $T$, for a decreasing sequence $\epsilon_i \to 0$. Such strategies
exist by Lemma~\ref{lem:weak-event-uniform}, both from the initial state and
from the set $T$. However, the mass of probability that can be guaranteed to
be synchronized in~$T$ by the successive strategies is always smaller than $1$, 
and therefore we need to argue that the remaining masses of probability 
(of size $\epsilon_i$) can also get synchronized in $T$, and despite their possible 
shift with the main mass of probability.

Two main
key arguments are needed to establish the correctness of the construction:
$(1)$ eventually synchronizing implies that a finite number of steps is sufficient
to obtain a probability mass of $1-\epsilon_i$ in $T$, and thus the construction
of the strategy is well defined, and $(2)$ by the finiteness of the period $r$
(such that $R = \Pre^r(R)$ where $R = \Pre^k(T)$ for some $k$) we can ensure 
to eventually make up for the shifts, 
and every piece of the probability mass can contribute
(synchronously) to the target infinitely often.

\begin{theorem}\label{theo:weakly-ls-is-as}
$\winlim{weakly}(\fsum_T) = \winas{weakly}(\fsum_T)$ for all MDPs and target sets~$T$.
\end{theorem}

\begin{proof}
Since $\winas{weakly}(\fsum_T) \subseteq \winlim{weakly}(\fsum_T)$ holds
by the definition, it is sufficient to prove that $\winlim{weakly}(\fsum_T) \subseteq \winas{weakly}(\fsum_T)$
and by Lemma~\ref{lem:weak-without-vanishing} it is sufficient to prove
that if $q_{\init} \in \winlim{weakly}(\fsum_T)$ is limit-sure weakly synchronizing 
in $T$ without vanishing state, then $q_{\init}$ is almost-sure weakly synchronizing 
in $T$. If $T$ has vanishing states, then consider $T' \subseteq T$ as in 
Lemma~\ref{lem:weak-without-vanishing} and it will follows that 
$q_{\init}$ is almost-sure weakly synchronizing in $T'$ and thus also in $T$.
We proceed with the proof that $q_{\init} \in \winlim{weakly}(\fsum_T)$ implies 
$q_{\init} \in \winas{weakly}(\fsum_T)$.

For $i=1,2,\dots$ consider the sequence of predecessors~$\Pre^{i}(T)$,
which is ultimately periodic: let $1 \leq k, r \leq 2^{\abs{Q}}$ such
that $\Pre^{k}(T) = \Pre^{k+r}(T)$, and let $R = \Pre^{k}(T)$.
Thus $R = \Pre^{k+r}(T) = \Pre^{r}(R)$.

\paragraph{Claim 1.} {\it We have $q_{\init} \in \winas{event}(\fsum_R)$ and 
$d_T \in \winas{event}(\fsum_R)$.}

\paragraph{Proof of Claim 1.}
By Lemma~\ref{lem:weak-event-uniform}, since there is no vanishing state in $T$
we have $q_{\init} \in \winlim{event}(\fsum_{\Pre(T)})$
and $d_T \in \winlim{event}(\fsum_{\Pre(T)})$, 
and it follows from the characterization of~\cite[Lemma~8]{DMS14} and the
proof of~\cite[Lemma~9]{DMS14} that: \medskip


\hspace{-4mm}
\begin{tabular}{ll}
either $(1)$ $q_{\init} \in \winsure{event}(\fsum_{\Pre(T)})$ or & $(2)$ $q_{\init} \in \winas{event}(\fsum_R)$, and  \\
either $(a)$ $d_T \in \winsure{event}(\fsum_{\Pre(T)})$ or & $(b)$ $d_T \in \winas{event}(\fsum_R)$.
\end{tabular} \medskip

Note that $(a)$ implies $(b)$ (and thus $(b)$ holds) since $(a)$ implies $T \subseteq \Pre^{i}(T)$
for some $i \geq 1$ (by~\cite[Lemma~4]{DMS14}) and thus $T \subseteq \Pre^{n \cdot i}(T)$
for all $n \geq 0$ by monotonicity of $\Pre^{i}(\cdot)$, which entails for 
$n \cdot i \geq k$ that $T \subseteq \Pre^m(R)$ where $m = (n \cdot i - k) \mod r$
and thus $d_T$ is sure (and almost-sure) winning for the eventually synchronizing 
objective in target~$R$.

Note also that $(1)$ implies $(2)$ since by $(1)$ we can play a sure-winning
strategy from $q_{\init}$ to ensure in finitely many steps probability~$1$ in $\Pre(T)$ 
and in the next step probability~$1$ in $T$, and by $(b)$ play an almost-sure 
winning strategy for eventually synchronizing in $R$. Hence $q_{\init} \in \winas{event}(\fsum_R)$
and thus $(2b)$ holds, which concludes the proof of Claim~1.

\paragraph{ } 
We now show that 
there exists an almost-sure winning strategy for the weakly synchronizing 
objective in target~$T$. 

Recall that $\Pre^r(R) = R$ and thus once some probability mass $p$ is in $R$,
it is possible to ensure that the probability mass in $R$ after $r$ steps
is at least $p$, and thus that (with period $r$) the probability in $R$ does not
decrease. 
By the result of~\cite[Lemma~9]{DMS14}, almost-sure winning for eventually 
synchronizing in $R$ implies that there exists a strategy $\alpha$ such that the 
probability in $R$ tends to $1$ at periodic positions: 
for some $0 \leq h < r$ the strategy $\alpha$ is \emph{almost-sure 
eventually synchronizing in $R$ with shift $h$}, that is 
$\forall \epsilon > 0 \cdot \exists N \cdot \forall n \geq N: n \equiv h \mod r \implies 
\M^{\alpha}_{n}(R) \geq 1 - \epsilon$. We also say that the initial distribution
$d_0 = \M^{\alpha}_{0}$ is almost-sure eventually synchronizing in $R$ with shift $h$.

\paragraph{Claim 2.} 
{\it
\begin{itemize}
\item[($\star$)] If $\M^{\alpha}_{0}$ is almost-sure eventually synchronizing 
in $R$ with some shift $h$, then $\M^{\alpha}_{i}$ 
is almost-sure eventually synchronizing in $R$ with shift $h-i \mod r$.
\item[($\star\star$)] Let $t$ such that $d_T$ is almost-sure eventually synchronizing in $R$ with shift $t$.
If a distribution is almost-sure eventually synchronizing in $R$ with some shift $h$, then 
it is also almost-sure eventually synchronizing in $R$ with shift $h + k + t \mod r$
(where we chose $k$ such that $R = \Pre^{k}(T)$).
\end{itemize}
}

\paragraph{Proof of Claim~2.}
The result ($\star$) immediately follows from the definition of shift, 
and we prove ($\star\star$) as follows.
We show that almost-sure eventually synchronizing in $R$ with shift $h$
implies almost-sure eventually synchronizing in $R$ with shift $h + k + t \mod r$. 
Intuitively, the probability mass
that is in $R$ with shift $h$ can be injected in $T$ in $k$ steps, and then from $T$
we can play an almost-sure eventually synchronizing strategy in target $R$ with shift $t$,
thus a total shift of $h + k + t \mod r$.
Precisely, an almost-sure winning strategy $\alpha$ is constructed as follows: 
given a finite prefix of play~$\rho$, if there is no state $q \in R$ that 
occurs in $\rho$ at a position $n \equiv h \mod r$, then play in $\rho$ according to 
the almost-sure winning strategy $\alpha_h$ for eventually synchronizing in $R$ with shift $h$. 
Otherwise, if there is no
$q \in T$ that occurs in $\rho$ at a position $n \equiv h+k \mod r$, then
we play according to a sure winning strategy $\alpha_{sure}$ for eventually synchronizing in $T$,
and otherwise we play according to an almost-sure winning strategy $\alpha_t$ from $T$ for 
eventually synchronizing in $R$ with shift $t$.
To show that $\alpha$ is almost-sure eventually synchronizing in $R$ with shift $h + k + t$,
note that $\alpha_h$ ensures with probability~$1$ that $R$ is reached at positions
$n \equiv h \mod r$, and thus $T$ is reached at positions $h+k \mod r$ by $\alpha_{sure}$, 
and from the states in $T$ the strategy $\alpha_t$ ensures with probability~$1$ 
that $R$ is reached at positions $h+k+t \mod r$. This concludes the proof of Claim~2.
\smallskip

\paragraph{Construction of an almost-sure winning strategy.}
We construct strategies $\alpha_\epsilon$ for $\epsilon > 0$ that ensure,
from a distribution that is almost-sure eventually 
synchronizing in $R$ (with some shift $h$), that after finitely many steps,
a distribution $d'$ is reached such that $d'(T) \geq 1-\epsilon$ and 
$d'$ is almost-sure eventually synchronizing in $R$ (with some shift $h'$).
Since $q_{\init}$ is almost-sure eventually synchronizing in $R$ (with some shift $h$),
it follows that the strategy $\alpha_{as}$ that plays successively 
the strategies (each for finitely many steps) $\alpha_{\frac{1}{2}}$, 
$\alpha_{\frac{1}{4}}$, $\alpha_{\frac{1}{8}}, \dots$ is 
almost-sure winning from $q_{\init}$ for the weakly synchronizing objective in target~$T$. 

We define the strategies $\alpha_\epsilon$ as follows.
Given an initial
distribution that is almost-sure eventually synchronizing in $R$ with a shift 
$h$ and given $\epsilon > 0$, let $\alpha_\epsilon$ be the strategy that plays
according to the almost-sure winning strategy $\alpha_h$ for eventually 
synchronizing in $R$ with shift $h$ for a number of steps $n \equiv h \mod r$ until a distribution 
$d$ is reached such that $d(R) \geq 1-\epsilon$, and then from $d$ it plays
according to a sure winning strategy $\alpha_{sure}$ for eventually synchronizing in $T$
from the states in $R$ (for $k$ steps), and keeps playing according to $\alpha_h$ from 
the states in $Q \setminus R$ (for $k$ steps). The distribution
$d'$ reached from $d$ after $k$ steps is such that $d'(T) \geq 1-\epsilon$ and we claim
that it is almost-sure eventually synchronizing in $R$ with shift $t$. 
This holds by definition from the states in $\Supp(d') \cap T$,
and by ($\star$) the states in $\Supp(d') \setminus T$ are almost-sure eventually synchronizing 
in $R$ with shift $h-(h+k) \mod r$, and by ($\star\star$) with shift $h-(h+k)+k+t = t$.

It follows that the strategy $\alpha_{as}$ is well-defined and ensures,
for all $\epsilon >0$, that the probability mass in $T$ is infinitely
often at least $1-\epsilon$, thus is almost-sure weakly synchronizing in $T$. 
This concludes the proof of Theorem~\ref{theo:weakly-ls-is-as}.
\qed
\end{proof}

The  complexity results of Theorem~\ref{theo:weakly-sure} and Theorem~\ref{theo:weakly-almost} 
hold for the membership problem with function $\fmax_T$ by the following lemma.

\begin{lemma}\label{lem:weakly-max-sum}
For weak synchronization and each winning mode, 
the membership problems with functions $\fmax$ and $\fmax_T$
are polynomial-time equivalent to the membership problem 
with function $\fsum_{T'}$ with a singleton ${T'}$.
\end{lemma}

\begin{proof}
First, for $\mu \in \{sure, almost, limit\}$,
we have $\win{weakly}{\mu}(\fmax_T) = \bigcup_{q \in T} \win{weakly}{\mu}(q)$,
showing that the membership problems for $\fmax$ 
are polynomial-time reducible to the corresponding membership problem 
for $\fsum_T$ with singleton~$T$. 
The reverse reduction is as follows. 
Given an MDP $\M$, a state $q$ and an initial distribution $d_0$,
we can construct an MDP $\M'$ and initial distribution $d'_0$
such that $d_0 \in \win{weakly}{\mu}(q)$ iff $d'_0 \in \win{weakly}{\mu}(\fmax_{Q'})$
where $Q'$ is the state space of $\M'$ (thus $\fmax_{Q'}$ is simply the function $\max$). 
The idea is to construct $\M'$ and $d'_0$
as a copy of~$\M$ and~$d_0$ where all states except $q$ are duplicated, and the 
initial and transition probabilities are equally distributed between 
the copies (see \figurename~\ref{fig:twin}). 
Therefore if the probability tends to~$1$ in some state,
it has to be in~$q$.

\begin{exclude}
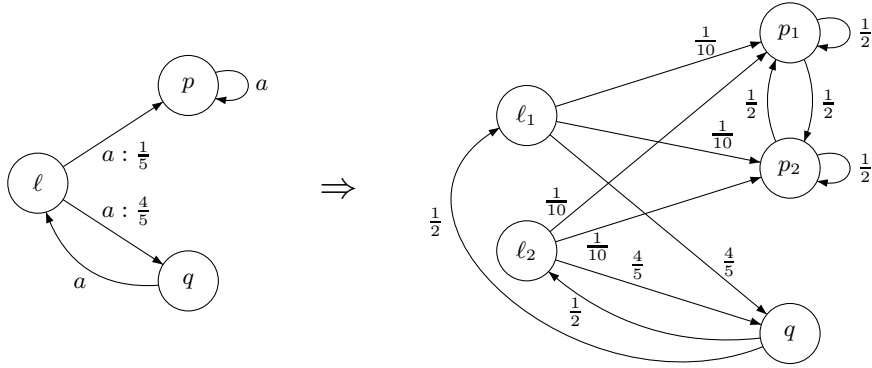
\begin{figure}[t]
\begin{center}
    \begin{picture}(120,50)(0,2)

\node[Nmarks=n](l)(5,25){$\ell$}
\node[Nmarks=n](p)(25,38){$p$}
\node[Nmarks=n](q)(25,12){$q$}

\drawedge[ELside=r,ELdist=0](l,p){$a:\frac{1}{5}$}
\drawedge[ELside=l,ELdist=0](l,q){$a:\frac{4}{5}$}

\drawedge[ELside=l, curvedepth=6](q,l){$a$}
\drawloop[ELside=l,loopCW=y, loopangle=0, loopdiam=4](p){$a$}

\node[Nframe=n](arrow)(45,24){{\Large $\Rightarrow$}}

\node[Nmarks=n](l1)(70,34){$\ell_1$}
\node[Nmarks=n](l2)(70,16){$\ell_2$}
\node[Nmarks=n](p1)(105,45){$p_1$}
\node[Nmarks=n](p2)(105,27){$p_2$}
\node[Nmarks=n](q)(105,5){$q$}

\drawedge[ELside=l,ELpos=70, ELdist=0](l1,p1){$\frac{1}{10}$}
\drawedge[ELside=l,ELpos=73, ELdist=.5](l1,p2){$\frac{1}{10}$}
\drawedge[ELside=l,ELpos=73, ELdist=0](l1,q){$\frac{4}{5}$}

\drawedge[ELside=l,ELpos=16, ELdist=0](l2,p1){$\frac{1}{10}$}
\drawedge[ELside=r,ELpos=25, ELdist=0](l2,p2){$\frac{1}{10}$}
\drawedge[ELside=l,ELpos=40, ELdist=.5](l2,q){$\frac{4}{5}$}

\drawbpedge[ELpos=70,ELside=l](q,210,30,l1,200,30){$\frac{1}{2}$} 
\drawedge[ELside=l, ELpos=75, ELdist=0, curvedepth=5](q,l2){$\frac{1}{2}$}

\drawloop[ELside=l,loopCW=y, loopangle=0, loopdiam=4](p1){$\frac{1}{2}$}
\drawedge[ELside=l,ELpos=52, curvedepth=3](p1,p2){$\frac{1}{2}$}
\drawedge[ELside=l,ELpos=48, curvedepth=3](p2,p1){$\frac{1}{2}$}
\drawloop[ELside=l,loopCW=y, loopangle=0, loopdiam=4](p2){$\frac{1}{2}$}

\end{picture}
\end{center}
 \caption{
State duplication ensures that the  probability mass can never 
be accumulated in a single state except in~$q$ (we omit action $a$ for readability).
\label{fig:twin}}
\end{figure}
\end{exclude}

\qed
\end{proof}

\section{Strong Synchronization}\label{sec:strongly-synch}

In this section, we show that the membership problem
for strongly synchronizing objectives can be solved in polynomial time,
for all winning modes, and both with function  $\fmax_T$ and
function $\fsum_T$. We show that linear-size memory is necessary in general
for $\fmax_T$, and memoryless strategies are sufficient for $\fsum_T$.
It follows from our results that the limit-sure and almost-sure winning
modes coincide for strong synchronization.

\subsection{Strong synchronization with function $\fmax$}

First, note
that for strong synchronization the membership problem with function $\fmax_T$ 
reduces to the membership problem with function $\fmax_Q$ where $Q$
is the entire state space, by a construction similar to the proof of 
Lemma~\ref{lem:weakly-max-sum}: states in $Q \setminus T$ are duplicated,
ensuring that only states in $T$ are used to accumulate probability.

The strongly synchronizing objective with function $\fmax$ requires that from 
some point on, almost all the probability mass is at every step in a single state.
The sequence of states that contain almost all the probability 
corresponds to a sequence of deterministic transitions in the MDP, and thus 
eventually to a cycle of deterministic transitions. 

The \emph{graph of deterministic transitions} of an MDP $\M = \tuple{Q,\Act,\delta}$
is the directed graph $G = \tuple{Q,E}$ where 
$E = \{ \tuple{q_1,q_2} \mid \exists a \in \Act: \delta(q_1,a)(q_2) = 1 \}$. 
For $\ell \geq 1$, a \emph{deterministic cycle in $\M$} of length $\ell$ is a finite path 
$\q_{\ell} \q_{\ell-1} \cdots \q_0$ in $G$ (that is, $\tuple{\q_{i},\q_{i-1}} \in E$
for all $1 \leq i \leq \ell$) such that $\q_0 = \q_{\ell}$. 
The cycle is \emph{simple} if $\q_i \neq \q_j$ for all $1 \leq i<j \leq \ell$.


We show that sure (resp., almost-sure and limit-sure) strong synchronization 
is equivalent to sure (resp., almost-sure and limit-sure) reachability to a 
state in such a cycle, with the requirement that it can be reached in a 
synchronized way, that is by finite paths whose lengths are congruent modulo 
the length $\ell$ of the cycle. To check this, we keep track of a modulo-$\ell$ counter
along the play.

Define the MDP $\M \times [\ell] = \tuple{Q',\Act,\delta'}$
where $Q' = Q \times \{0,1,\cdots,\ell-1\}$ and 
$\delta'(\tuple{q,i},a)(\tuple{q',i-1}) = \delta(q,a)(q')$ 
(where $i-1$ is $\ell-1$ for $i=0$)
for all states $q,q'\in Q$, actions $a \in \Act$, and $0 \leq i \leq \ell-1$.

\begin{lemma}\label{lem: p-strong}
Let $\eta$ be the smallest positive probability in the transitions of $\M$,
and let $\frac{1}{1+\eta}< p \leq 1$. 
There exists a strategy~$\alpha$ such that $\liminf_{n \to \infty} \norm{\M^{\alpha}_n}\geq p$
from an initial state~$q_{\init}$
if and only if 
there exists a simple deterministic cycle~$\q_{\ell} \q_{\ell-1} \cdots \q_0$ in $\M$ 
and a strategy~$\beta$ in~$\M\times[\ell]$ such that $\Pr^{\beta}(\Diamond \{\tuple{\q_0,0}\}) \geq p$
from~$\tuple{q_{\init},0}$.
\end{lemma}

\begin{proof}
First, assume that there exists a simple deterministic cycle~$\q_{\ell} \q_{\ell-1} \cdots \q_0$
with length~$\ell$ and a strategy~$\beta$ in $\M \times [\ell]$ that ensures the target set  
$\Diamond \{\tuple{\q_0,0}\}$ is reached with probability at least $p$  
from the state~$\tuple{q_{\init},0}$. Since randomization is not necessary for 
reachability objectives, we can assume that $\beta$ is a pure strategy.
We show that 
there exists a strategy~$\alpha$ such that $\liminf_{n \to \infty} \norm{\M^{\alpha}_n}\geq p$
from~$q_{\init}$.
From $\beta$, we construct  a pure strategy $\alpha$ in $\M$. 
Given a finite path $\rho = q_0 a_0 q_1 a_1\dots q_n$ in $\M$ (with $q_0 = q_{\init}$), there 
is a corresponding path $\rho' = \tuple{q_0,k_0}a_0 \tuple{q_1,k_1} a_1 \dots \tuple{q_n, k_n}$ in 
$\M \times [\ell]$ where $k_i = -i~mod~\ell$.
Since the sequence $k_0 k_1 \dots$ is uniquely determined from $\rho$, 
there is a clear bijection between the paths in $\M$ and the paths in $\M \times [\ell]$
that we often omit to apply and mention.
For $\rho$, we define $\alpha$ as follows: 
if $q_n = \q_{k_n}$, then there exists an action $a$ such that 
$\post(\q_{k_n},a)=\{\q_{k_{n+1}}\}=\{\q_{n+1}\}$ and we define $\alpha(\rho) = a$, 
otherwise let $\alpha(\rho) = \beta(\rho')$.	
Thus $\alpha$ mimics $\beta$  unless a 
state $q$ is reached at step $n$ such that 
$q =\q_{k}$ where $k=-n~mod~\ell$, 
and then $\alpha$ switches to always playing actions
that keeps $\M$ in the simple deterministic cycle~$\q_{\ell} \q_{\ell-1} \cdots \q_0$.
Below, we prove that 
given $\epsilon>0$ there exists $k$ such that 
for all $n\geq k$, we have $\norm{\M^{\alpha}_{n}} \geq p-\epsilon$.
It follows that $\liminf_{n \to \infty} \norm{\M^{\alpha}_n}\geq p$ from $q_{\init}$.
Since $\Pr^{\beta}(\Diamond \{\tuple{\q_0,0}\})\geq p$,
there exists $k$ such that 
$\Pr^{\beta}(\Diamond^{\leq k}  \{\tuple{\q_0,0}\}) \geq p-\epsilon$.
We assume w.l.o.g. that $k~mod~\ell = 0$.
For $i = 0, 1, \dots, \ell-1$,
let $R_i = \{\tuple{\q_{i},i}\}$. 
Then trivially $\Pr^{\beta}(\Diamond^{\leq k} \bigcup_{i=0}^{\ell-1} R_i) \geq p-\epsilon$
and since $\alpha$ agrees with $\beta$ on all finite paths  
that do not (yet) visit $\bigcup_{i=0}^{\ell-1} R_i$, given a path $\rho$ that
visits $\bigcup_{i=0}^{\ell-1} R_i$ (for the first time), only  actions that keep $\M$
 in the simple cycle $\q_{\ell} \q_{\ell-1} \cdots \q_{0}$
are played by $\alpha$ and thus all continuations of $\rho$ in the outcome of $\alpha$ 
will visit $\q_0$ after $k$ steps (in total).  
It follows that $\Pr^{\beta}(\Diamond^{k} \{\tuple{\q_0,0}\}) \geq p-\epsilon$, that is 
$\M^{\alpha}_k(\q_0) \geq p-\epsilon$. Thus, $\norm{\M^{\alpha}_k}\geq p-\epsilon$. Since
next, $\alpha$ will always play actions that keeps  $\M$ looping through the  cycle 
$\q_{\ell} \q_{\ell-1} \cdots \q_0$, we have $\norm{\M^{\alpha}_n}\geq p-\epsilon$ for all $n\geq k$. 

Second, assume that there exists a strategy~$\alpha$ such that 
$\liminf_{n \to \infty} \norm{\M^{\alpha}_n}\geq p$
from~$q_{\init}$.
Thus, for all $\epsilon>0$ there exists $k\in \nat$ such that  for all $n \geq k$ 
we have $\norm{\M^{\alpha}_{n}} \geq p-\epsilon$.
Fix $\epsilon<p-\frac{1}{1+\eta}$. 
Let $k$ be such that for all $n \geq k$, there 
exists a unique state $\p_n$ such that  $\M^{\alpha}_n(\p_n) \geq p-\epsilon$.
Below, we prove that  for all $n \geq k$, there exists some action $a\in \Act$ 
such that $\post(\p_{n},a)=\{\p_{n+1}\}$.  
Assume towards contradiction that there exists $j> k$ such that 
for all~$a$ there exists~$q\neq \p_{j+1}$ such that $\{q,\p_{j+1}\} \subseteq \post(\p_{j},a)$.  
Therefore, $\M^{\alpha}_{j+1}(q)\geq \M^{\alpha}_{j}(\p_j) \cdot \eta \geq (p- \epsilon) \cdot \eta$. 
Hence,  $$\M^{\alpha}_{j+1}(\p_{j+1})\leq 1- \M^{\alpha}_{j+1}(q)\leq 1-(p- \epsilon)\cdot \eta.$$ 
 Thus,
$p- \epsilon\leq \norm{\M^{\alpha}_{j+1}} \leq  1-(p- \epsilon)\cdot \eta$ that 
 gives $\epsilon\geq p-\frac{1}{1+\eta}$,  a contradiction.
This argument proves that for all $n\geq k$, there exists an action $a\in \Act$ 
such that $\post(\p_{n},a)=\{\p_{n+1}\}$. 
The finiteness of the state space~$Q$ entails that 
in the sequence $\p_k \p_{k+1} \cdots$, 
some state  and thus some simple deterministic cycle
occur infinitely often.
Let $\q_{\ell} \q_{\ell-1} \cdots \q_0$ be a cycle that occurs  infinitely often 
in the sequence $\p_k\p_{k+1} \cdots$ (in the right order).
For all $j$, let $i_j$ be the position of $\q_0$ in all occurrences of the cycle 
$\q_{\ell} \q_{\ell-1} \cdots \q_0$ in the sequence $\p_k \p_{k+1} \cdots$;
 and let $t_j=i_j~mod~\ell$. In the sequence $t_0 t_1 \cdots$,
there exists $0\leq t<\ell$ that appears infinitely often.
Let the cycle $r_{\ell} r_{\ell-1} \cdots r_0$ be such that
 $r_{(i+t)~mod~\ell}=\q_i$. Then, the cycle $r_{\ell} r_{\ell-1} \cdots r_0$
happens infinitely often in the sequence $\p_k\p_{k+1} \cdots$
such that the positions of $r_0$ are infinitely often $0$ (modulo~$\ell)$.
Therefore, the probability of $\M$ to be in $r_0$ in positions (modulo~$\ell$) equals to~$0$,
is infinitely often equal or greater than $p$.
Hence, for a  
strategy $\beta$ in $\M\times[\ell]$ that copies
all the plays of the strategy $\alpha$, we have $\Pr^{\beta}(\Diamond\{\tuple{r_0,0}\})\geq p$ from~$\tuple{q_{\init},0}$. \qed
\end{proof}

It follows directly from Lemma~\ref{lem: p-strong} with $p=1$ that 
almost-sure strong synchronization is equivalent to almost-sure
reachability to a deterministic cycle in~$\M\times[\ell]$. The same equivalence
holds for the sure and limit-sure winning modes.

\begin{lemma}\label{lem: strong}
A state $q_{\init}$ is sure (resp., almost-sure or limit-sure) winning
for the strongly synchronizing objective (according to $\fmax_Q$)
if and only if there exists a simple deterministic cycle~$\q_{\ell} \q_{\ell-1} \cdots \q_0$ 
such that $\tuple{q_{\init},0}$ is sure (resp., almost-sure or limit-sure) 
winning for the reachability objective~$\Diamond \{\tuple{\q_0,0}\}$
in $\M\times[\ell]$.
\end{lemma}

\begin{proof}
The proof is organized in three sections:

{\bf (1) sure winning mode:}
First, assume that  there exists a simple deterministic  
cycle~$\q_{\ell} \q_{\ell-1} \cdots \q_0$ with length~$\ell$ 
 such that $\tuple{q_{\init},0}$ 
is  sure winning for the reachability objective~$\Diamond \{\tuple{\q_0,0}\}$.
Thus, there exists a pure memoryless strategy $\beta$
such that $\Outcomes(\tuple{q_{\init},0}, \beta) \subseteq \Diamond \{\tuple{\q_0,0}\}$.
Since  $\beta$ is memoryless,  there must be $k\leq \abs{Q}\times \ell$ such that 
$\Outcomes(\tuple{q_{\init},0}, \beta)\subseteq \Diamond^{\leq k} \{\tuple{\q_0,0}\}$
meaning that all infinite paths starting in $\tuple{q_{\init},0}$ and following~$\beta$
reach  $\tuple{\q_0,0}$ within~$k$ steps.
From $\beta$, we construct  a pure finite-memory strategy $\alpha$ in $\M$ that is 
represented by $T=\tuple{\mem,i, \alpha_u, \alpha_n}$ 
where $\mem=\{0,\cdots, \ell-1\}$ is the set of modes.
The idea is that $\alpha$ simulates what $\beta$ plays in the state~$\tuple{q,i}$,
in the state $q$ of~$\M$ and mode~$i$ of~$T$ (there is only one exception).
Thus, the  initial mode is~$0$. The update function 
 only decreases modes by~$1$ ($\alpha_u(i,a, q)=(i-1)~mod~\ell$ 
for all states~$q$ and actions~$a$) since by 
taking any transition the 
mode is decreased by~$1$.  
The next-move function $\alpha_n(i,q)$ is defined 
as follows: $\alpha_n(i,q)=\beta(\tuple{q,i})$ for all states~$q$ and modes $0\leq i<\ell$, 
except when  $q= \q_i$, in this case let $\alpha_n(i,q)=a$ where
$\post(\q_i,a) =\{q_{i-1}\}$.
Thus $\beta$ mimics $\alpha$  unless a 
state $q$ is reached at step $n$ such that $q =\q_{-n~mod~\ell}$, 
and then $\alpha$ switches to always playing actions
that keeps $\M$ in the simple deterministic cycle~$\q_{\ell} \q_{\ell-1} \cdots \q_0$.
Now we prove that 
 $q_{\init}$ is sure winning for the strongly synchronizing objective according to $\fmax_Q$.
Let $j\geq k$ be such that $j~mod~\ell = 0$.
Let $R = \{\tuple{\q_i,i}\mid 0\leq i<\ell\}$. 
Thus obviously  $\Outcomes(\tuple{q_{\init},0}, \beta) \subseteq \Diamond  R$. 
and since $\alpha$ agrees with $\beta$ on all finite paths  
that do not (yet) visit $R$, given a path $\rho$ that
visits $R$ (for the first time), only  actions that keep $\M$
 in the simple cycle $\q_{\ell} \q_{\ell-1} \cdots \q_0$
are played by $\alpha$ and thus all continuations of $\rho$ in the outcome of $\alpha$ 
will visit $\q_0$ after $j$ steps.  
It follows that $\Pr^{\beta}(\Diamond^{j} \{\tuple{\q_0,0}\}) =1$, that is 
$\M^{\alpha}_j(q_0)=1$. Thus, $\norm{\M^{\alpha}_j}=1$. Since
next, $\alpha$ will always play actions that keeps  $\M$ looping through the  cycle 
$\q_{\ell} \q_{\ell-1} \cdots \q_0$, we have $\norm{\M^{\alpha}_n}=1$ for all $n\geq j$.

Second, assume that there exists a strategy~$\alpha$ and $k$ 
such that for all $n\geq k$ we have $\norm{\M^{\alpha}_n}=1$ 
from the initial state~$q_{\init}$.
For all~$n\geq k$, let $\p_n$ be a state such that  $\M^{\alpha}_n(\p_n)=1$. 
The finiteness of the state space~$Q$ entails that 
in the sequence $\p_k \p_{k+1} \cdots$, 
some state  and thus some simple deterministic cycle
occur infinitely often.
Let $\q_{\ell} \q_{\ell-1} \cdots \q_0$ be a cycle that occurs  infinitely often 
in the sequence $\p_k\p_{k+1} \cdots$ (in the right order).
For all $j$, let $i_j$ be the position of $\q_0$ in all occurrences of the cycle 
$\q_{\ell} \q_{\ell-1} \cdots \q_0$ in the sequence $\p_k \p_{k+1} \cdots$;
 and let $t_j=i_j~mod~\ell$. In the sequence $t_0 t_1 \cdots$,
there exists $0\leq t<\ell$ that appears infinitely often.
Let the cycle $r_{\ell} r_{\ell-1} \cdots r_0$ be such that
 $r_{(i+t)~mod~\ell}=\q_i$. Then, the cycle $r_{\ell} r_{\ell-1} \cdots r_0$
happens infinitely often in the sequence $\p_k\p_{k+1} \cdots$
such that the positions of $r_0$ are infinitely often $0$ (modulo~$\ell)$.
Hence, for a  
strategy $\beta$ in $\M\times[\ell]$ that copies
all the plays of the strategy $\alpha$, 
we have $\Outcomes(\tuple{q_{\init},0}, \beta) \subseteq \Diamond \{\tuple{r_0,0}\}$ from
the initial state~$\tuple{q_{\init},0}$.

{\bf (2) almost-sure winning mode:} 
This case is an immediate result from Lemma~\ref{lem: p-strong}, by 
taking $p=1$.

{\bf (3) limit-sure winning mode:} First, assume that there exists a simple deterministic  
cycle~$\q_{\ell} \q_{\ell-1} \cdots \q_0$ with length~$\ell$ 
 such that $\tuple{q_{\init},0}$ is  limit-sure (and thus almost-sure) winning for the 
reachability objective~$\Diamond \{\tuple{\q_0,0}\})$. 
Since $\tuple{\q_{\init}, 0}$ is almost-sure for reachability objective,
then $q_{\init}$ is almost-sure (and thus limit-sure) for strongly synchronizing objective.
Second, assume that   $q_{\init}$ is limit-sure
winning for the strongly synchronizing objective (according to $\fmax_Q$).
It means that for all $i$ there exists a strategy $\alpha_i$ such that 
$\liminf_{n\to\infty} \norm{\M^{\alpha_i}_n}\geq 1-2^{-i}$.
Let $k$ be such that $1-2^{-k}\geq \frac{1}{1+\eta}$.
By Lemma~\ref{lem: p-strong},
for all~$i\geq k$
 there exists a simple deterministic  cycle~$c_i = \p_{\ell_i} \p_{\ell_i-1} \cdots \p_0$ with length~$\ell_i$ 
and a strategy~$\beta_i$ in~$\M\times[\ell_i]$ such that 
$\Pr^{\beta_i}(\Diamond \{\tuple{\q_0,0}\}) \geq 1-2^{-i}$
 from~$\tuple{q_{\init},0}$.
Since the number of simple deterministic cycle is finite,
there exists some simple cycle $c$ that occurs infinitely often in the
sequence $c_k c_{k+1} c_{k+2}\cdots$. We see that 
for the cycle $c = \q_{\ell} \q_{\ell-1} \cdots \q_0$, the states 
$\tuple{\q_{\init},0}$ is limit-sure winning for
the reachability objective~$\Diamond \{\tuple{\q_0,0}\})$.
\qed
\end{proof}

Since the winning regions of almost-sure and limit-sure winning
coincide for reachability objectives in MDPs~\cite{AHK07},
the next corollary follows from  Lemma~\ref{lem: strong}. 

\begin{corollary}\label{col: almost-limit-strong}
$\winlim{strongly}(\fmax_T)= \winas{strongly}(\fmax_T)$
for all target sets~$T$.
\end{corollary}

If there exists a cycle $c$ satisfying the condition in Lemma~\ref{lem: strong},
then all cycles reachable from $c$ in the graph $G$ of deterministic transitions
also satisfies the condition. Hence it is sufficient to check the condition
for an arbitrary simple cycle in each strongly connected component (SCC)
of $G$. It follows that strong synchronization can be 
decided in polynomial time (SCC decomposition can be computed in polynomial 
time, as well as sure, limit-sure, and almost-sure reachability in MDPs).
The length of the cycle gives a linear bound on the memory needed to win, and 
the bound is tight.

\begin{theorem}\label{theo:strongly-max}
For the three winning modes of strong synchronization 
according to $\fmax_T$ in MDPs:

\begin{enumerate}
\item (Complexity). The membership problem is PTIME-complete.

\item (Memory). Linear memory is necessary and sufficient for both pure 
and randomized strategies, and pure strategies are sufficient.
\end{enumerate}
\end{theorem}

\begin{proof}
First, we prove the PTIME upper bound.
Given an MDP~$\M=\tuple{Q,\Act,\delta}$ and a state~$q_{\init}$, we say a simple deterministic 
cycle~$c=\q_{\ell} \q_{\ell-1} \cdots \q_0$ is sure winning (resp., almost-sure
 and limit-sure) for strong synchronization from~$q_{\init}$ if $\tuple{q_{\init},0}$
is sure winning  (resp., almost-sure and limit-sure) 
for the reachability objective~$\Diamond \{ \tuple{\q_0,0}\}$ in $\M\times[\ell]$.
We show that  
if  $c$ is sure winning 
(resp., almost-sure and limit-sure) for strong synchronization from~$q_{\init}$,
then so are all
  simple cycles
$c'=\p_{\ell'} \p_{\ell'-1} \cdots \p_0$  reachable from $c$
in the deterministic digraph  induced by $\M$.  

{\bf (1) sure winning:} Since $c$ is sure winning for  strong synchronization from~$q_{\init}$,
 $\M$ is $1$-synchronized in $\q_0$. 
Since there is a path via deterministic transitions
from $\q_0$ to $\p_0$,  $\M$ is $1$-synchronized in $\p_0$ too. 
So the cycle $c'$ is 
sure wining for  strong synchronization from~$q_{\init}$, too.

{\bf (2) limit-sure winning:} Assume that $c$ is limit-sure winning for  strong 
synchronization from~$q_{\init}$. 
By definition, for all $i\in \nat$, there exists $n$ such that  for 
all $j>n$ we have $\M$ is $1-2^{-i-j}$ in $\q_0$. It implies that 
for all $i$ there exists $n$ such that $\M$ is $1-2^{-2i}$-synchronized in $\q_0$.
Since there is a path via deterministic transitions
from $\q_0$ to $\p_0$, then $\M$ is $1-2^{-2i}$-synchronized in $\p_0$ for all~$i$. So the cycle $c'$ is 
limit-sure wining for  strong synchronization from~$q_{\init}$, too.

{\bf (3) almost-sure winning:} By corollary~\ref{col: almost-limit-strong}, 
since a cycle is almost-sure winning for  strong 
synchronization from~$q_{\init}$ if and only if it is limit-sure winning, the results follows.

The above arguments prove that 
if  a simple deterministic cycle $c$ is
 sure winning (resp., almost-sure and limit-sure)
 for strongly synchronizing objective from~$q_{\init}$, 
then all simple cycles reachable 
from $c$ in the graph of deterministic transitions~$G$ induced by $\M$, are 
sure winning (resp., almost-sure and limit-sure). 
In particular, it holds for all
simple cycles in the bottoms SCCs reachable
from $c$ in~$G$. 
Therefore, to decide membership problem for strongly synchronizing objective, 
it suffices to only check whether one cycle in each bottom SCC of $G$ is sure winning
(resp., almost-sure and limit-sure). 
Since the SCC decomposition for a digraph is in PTIME, 
and since the number of bottom SCCs in a connected digraph
is at most the size of the digraph (the number of states $\abs{Q}$),
the PTIME upper bound follows.  

For the PTIME-hardness, for all $\mu\in \{sure, almost, limit\}$ the proof is
by a reduction from monotone Boolean circuit problem (MBC).
Given an MBC with an underlying binary tree, 
the value of leaves are either $0$ or $1$ (called $0$ or $1$-leaves),
and the value of other vertices,  labeled
with $\wedge$ or $\vee$, are computed inductively.
It is known that deciding whether the value of the \emph{root}
is $1$ for a given  MBC,  is PTIME-complete~\cite{GB88}.
From an MBC, we construct an MDP $\M$ where the states are the 
vertices of the tree with three new absorbing states
$\sync$, $q_1$ and $q_2$,  and two actions $L,R$. 
On both actions $L$ and $R$, the next successor of
the $1$-leaves is only $\sync$, and the next successor of
the $0$-leaves is  $q_1$ or $q_2$ with probability~$\frac{1}{2}$.
The next successor of
a $\wedge$-state is each of their children with equal probability,
on all actions.
The next successor of a $\vee$-state is its left (resp., right)
child by action $L$ (resp., action $R$).
We can see that $\M$ can be synchronized only in $\sync$.

We call a subtree \emph{complete} if 
($1$)  \emph{root} is in subtree,
($2$) at least one child of all $\vee$-vertices is
 in the subtree,
($3$) both children of all $\wedge$-vertices  are 
 in the subtree. 
There is a bijection between a complete
subtree and a strategy in~$\M$.
The value of \emph{root} is 1 if and only if
there is a complete subtree such that it has
no $0$-leaves (all leaves are $1$-leaves).
For such subtrees, all plays under the corresponding strategy  
reach some $1$-leave and thus are  synchronized in~$\sync$.
It means that $root \in \win{strongly}{\mu}(\sync)$
if and only if
  the value of $root$ is $1$.

Finally, the result on memory requirement is established as follows.
Since memoryless strategies are sufficient for reachability objectives
in MDPs, it follows from the proof of Lemma~\ref{lem: p-strong} and 
Lemma~\ref{lem: strong} that the (memoryless) winning strategies in $\M\times[\ell]$ 
can be transferred to winning strategies with memory $\{0,1,\cdots,\ell-1\}$ in $\M$.
Since $\ell \leq \abs{Q}$, linear-size memory is sufficient to win strongly 
synchronizing objectives. We present a family of MDPs $\M_n$ ($n \in \nat$)
that are sure winning for strongly synchronization (according to $\max_Q$), 
and where the sure winning strategies require linear memory. 
The MDP $\M_2$ is shown in \figurename~\ref{fig:strong-max-memory},
and the MDP $\M_n$ is obtained by replacing the cycle $q_2 q_3$ of deterministic
transitions by a simple cycle of length $n$. Note that only in $q_1$ there is a real
strategic choice. Since $q_1$ and $q_2$ contain probability, we need to wait
in $q_1$ (by playing $b$) until we play $a$ when the probability in $q_2$ 
comes back in $q_2$ through the cycle. We need to play $n-1$ times $b$, 
and then $a$, thus linear memory is sufficient and it is easy to show that
it is necessary to ensure strongly synchronization.
 \qed
\end{proof}

\begin{exclude}
\begin{figure}[t]
\begin{center}
    \begin{picture}(50,42)

\node[Nmarks=i](n0)(10,20){$q_{\init}$}
\node[Nmarks=n](n1)(25,30){$q_1$}
\node[Nmarks=n](n2)(25,10){$q_2$}
\node[Nmarks=n](n3)(45,10){$q_3$}

\drawpolygon[dash={0.8 0.5}0](20,18)(50,18)(50,2)(20,2)
\node[Nframe=n](label)(35,0){target set $T$}

\drawedge[ELpos=40, ELside=l](n0,n1){$a,b: \frac{1}{2}$}
\drawedge[ELpos=40, ELside=r](n0,n2){$a,b: \frac{1}{2}$}
\drawedge(n0,n2){}
\drawloop[ELside=l,loopCW=y, loopangle=90, loopdiam=4](n1){$b$}
\drawedge(n1,n2){$a$}

\drawedge[ELpos=50, ELside=l, curvedepth=4](n2,n3){$a,b$}
\drawedge[ELpos=50, ELside=l, curvedepth=4](n3,n2){$a,b$}

\end{picture}
\end{center}
 \caption{An MDP that all strategies to win sure strongly synchronizing with function $\fmax_{\{q_2,q_3\}}$ require memory.\label{fig:strong-max-memory}}
\end{figure}
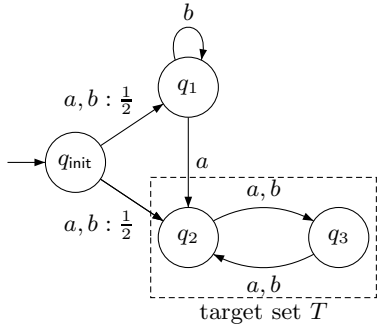
\end{exclude}

\subsection{Strong synchronization with function $\fsum$}

The strongly synchronizing objective with function $\fsum_{T}$ requires
that eventually all the probability mass remains in $T$. We show that this
is equivalent to a traditional reachability objective with target defined
by the set of sure winning initial distributions for the safety objective $\Box T$.

It follows that almost-sure (and limit-sure) winning for strong synchronization
is equivalent to almost-sure (or equivalently limit-sure) winning for
the coB\"uchi objective $\Diamond \Box T = \{q_{0} a_{0} q_{1} \dots \in \Plays(\M) 
\mid \exists j \cdot \forall i > j: q_i \in T\}$. However, sure strong synchronization
is not equivalent to sure winning for the coB\"uchi objective:
the MDP in \figurename~\ref{fig:coBuchi} is sure winning for the coB\"uchi objective $\Diamond \Box \{q_{\init},q_2\}$
from $q_{\init}$, but not sure winning for the reachability objective $\Diamond S$
where $S = \{q_2\}$ is the winning region for the safety objective $\Box \{q_{\init},q_2\}$
(and thus not sure strongly synchronizing). Note that this MDP is almost-sure
strongly synchronizing in target $T = \{q_{\init},q_2\}$ from $q_{\init}$, and almost-sure
winning for the coB\"uchi objective $\Diamond \Box T$, as well as almost-sure
winning for the reachability objective $\Diamond S$.

\begin{figure}[t]
\begin{center}
    \begin{picture}(60,18)(0,2)

\node[Nmarks=ir](q0)(5,5){$q_{\init}$}
\node[Nmarks=n](q1)(30,5){$q_1$}
\node[Nmarks=r](q2)(55,5){$q_2$}

\drawloop[ELside=l, loopangle=90, loopdiam=4](q0){$a:\frac{1}{2}$}
\drawedge[ELside=l, ELdist=.5](q0,q1){$a:\frac{1}{2}$}
\drawedge[ELside=l, ELdist=1](q1,q2){$a$}
\drawloop[ELside=l, loopangle=90, loopdiam=4](q2){$a$}

\end{picture}
\end{center}
 \caption{An MDP such that $q_{\init}$ is sure-winning for 
coB\"uchi objective in $T=\{q_{\init},q_2\}$ but not for 
strong synchronization according to $\fsum_T$. 
\label{fig:coBuchi}}
\end{figure}
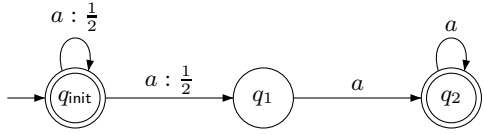

\begin{lemma}\label{lem: strong-sum}
Given a target set~$T$, an MDP $\M$ is sure (resp., almost-sure or limit-sure) winning 
for the strongly synchronizing objective according to $\fsum_T$
if and only if $\M$ is sure (resp., almost-sure or limit-sure) winning 
for the reachability objective~$\Diamond S$  
where $S$ is the sure winning region for the safety objective~$\Box T$.
\end{lemma}

\begin{proof}
First, assume that a state $q_{\init}$ of $\M$ is sure (resp., almost-sure 
or limit-sure) winning for the strongly synchronizing objective according to $\fsum_T$,
and show that $q_{\init}$ is sure (resp., almost-sure or limit-sure) 
winning for the reachability objective~$\Diamond S$.

$(i)$ \emph{Limit-sure winning}. For all $\epsilon >0$, 
let $\epsilon' = \frac{\epsilon}{\abs{Q}} \cdot \eta^{\abs{Q}}$ where
$\eta$ is the smallest positive probability in the transitions of $\M$.
By the assumption, from $q_{\init}$ there exists a strategy~$\alpha$ and $N \in \nat$
such that for all $n \geq N$, we have $\M^{\alpha}_n(T) \geq 1-\epsilon'$.
We claim that at step $N$, all non-safe states have probability at most $\frac{\epsilon}{\abs{Q}}$,
that is $\M^{\alpha}_N(q) \leq \frac{\epsilon}{\abs{Q}}$ for all $q \in Q \setminus S$.
Towards contradiction, assume that $\M^{\alpha}_N(q) > \frac{\epsilon}{\abs{Q}}$
for some non-safe state $q \in Q \setminus S$. Since $q \not\in S$ is not safe, 
there is a path of length $\ell \leq \abs{Q}$ from $q$ to a state in $Q \setminus T$,
thus with probability at least $\eta^{\abs{Q}}$. It follows that after $N + \ell$ steps we have
$\M^{\alpha}_{N+\ell}(Q \setminus T) > \frac{\epsilon}{\abs{Q}} \cdot \eta^{\abs{Q}} = \epsilon'$,
in contradiction with the fact $\M^{\alpha}_n(T) \geq 1-\epsilon'$ for all $n \geq N$.
Now, since all non-safe states have probability at most $\frac{\epsilon}{\abs{Q}}$
at step $N$, it follows that $\M^{\alpha}_N(Q \setminus S) \leq \frac{\epsilon}{\abs{Q}}\cdot \abs{Q} = \epsilon$
and thus $\Pr^{\alpha}(\Diamond S) \geq 1-\epsilon$. Therefore $\M$ is limit-sure
winning for the reachability objective $\Diamond S$ from $q_{\init}$.

$(ii)$ \emph{Almost-sure winning}. Since almost-sure strongly synchronizing
implies limit-sure strongly synchronizing, it follows from $(i)$ that
$\M$ is limit-sure (and thus also almost-sure) winning for the reachability objective 
$\Diamond S$, as limit-sure and almost-sure reachability coincide for MDPs~\cite{AHK07}.

$(iii)$ \emph{Sure winning}. From $q_{\init}$ there exists a strategy~$\alpha$ and $N \in \nat$ such 
that for all $n \geq N$, we have $\M^{\alpha}_n(T)=1$. Hence $\alpha$ is sure
winning for the reachability objective $\Diamond \Supp(\M^{\alpha}_N)$,
and from all states in $\Supp(\M^{\alpha}_N)$ the strategy $\alpha$ ensures
that only states in $T$ are visited. It follows that $\Supp(\M^{\alpha}_N) \subseteq S$
is sure winning for the safety objective $\Box T$, and thus $\alpha$ is sure winning
for the reachability objective $\Diamond S$ from $q_{\init}$.

For the converse direction of the lemma, assume that a state $q_{\init}$ is sure (resp., almost-sure 
or limit-sure) winning for the reachability objective~$\Diamond S$.
We construct a winning strategy for strong synchronization in $T$ as follows: 
play according to a sure (resp., almost-sure or limit-sure) winning
strategy for the reachability objective~$\Diamond S$, and whenever a state
of $S$ is reached along the play, then switch to a winning
strategy for the safety objective~$\Box T$. The constructed 
strategy is sure (resp., almost-sure or limit-sure) winning
for strong synchronization according to $\fsum_T$ because for
sure winning, after finitely many steps all paths from $q_{\init}$
end up in $S \subseteq T$ and stay in $S$ forever, and for almost-sure
(or equivalently limit-sure) winning, for all $\epsilon > 0$, after sufficiently 
many steps, the set $S$ is reached with probability at least $1 - \epsilon$,
showing that the outcome is strongly ($1 - \epsilon$)-synchronizing in $S \subseteq T$,
thus the strategy is almost-sure (and also limit-sure) strongly synchronizing.
\qed
\end{proof}

\begin{corollary}\label{col: almost-limit-strong-sum}
$\winlim{strongly}(\fsum_T) = \winas{strongly}(\fsum_T)$
for all target sets~$T$.
\end{corollary}

The following result follows from Lemma~\ref{lem: strong-sum} and the fact 
that the winning region for sure safety, sure reachability, and almost-sure 
reachability can be computed in polynomial time for MDPs~\cite{AHK07}. 
Moreover, memoryless strategies are sufficient for these objectives.

\begin{theorem}\label{theo:strongly-sum}
For the three winning modes of strong synchronization 
according to $\fsum_T$ in MDPs:

\begin{enumerate}
\item (Complexity). The membership problem is PTIME-complete.

\item (Memory). Pure memoryless strategies are sufficient.
\end{enumerate}
\end{theorem}
\bibliographystyle{splncs03}
\bibliography{biblio} 

\begin{thebibliography}{10}
\providecommand{\url}[1]{\texttt{#1}}
\providecommand{\urlprefix}{URL }

\bibitem{AAGT12}
Agrawal, M., Akshay, S., Genest, B., Thiagarajan, P.S.: Approximate
  verification of the symbolic dynamics of {M}arkov chains. In: Proc. of LICS.
  pp. 55--64. IEEE (2012)

\bibitem{ConcOmRegGames}
de~Alfaro, L., Henzinger, T.A.: Concurrent omega-regular games. In: Proc. of
  LICS. pp. 141--154 (2000)

\bibitem{AHK07}
de~Alfaro, L., Henzinger, T.A., Kupferman, O.: Concurrent reachability games.
  Theor. Comput. Sci.  386(3),  188--217 (2007)

\bibitem{BS96}
Bach, E., Shallit, J.: Algorithmic Number Theory, Vol. 1: Efficient Algorithms.
  MIT Press (1996)

\bibitem{BBG08}
Baier, C., Bertrand, N., Gr{\"o}{\ss}er, M.: On decision problems for
  probabilistic {B}{\"u}chi automata. In: Proc. of FoSSaCS. pp. 287--301. LNCS
  4962, Springer (2008)

\bibitem{BBMR08}
Baldoni, R., Bonnet, F., Milani, A., Raynal, M.: On the solvability of
  anonymous partial grids exploration by mobile robots. In: Proc. of OPODIS.
  pp. 428--445. LNCS 5401, Springer (2008)

\bibitem{Burkhard76a}
Burkhard, H.D.: Zum l{\"a}ngenproblem homogener experimente an determinierten
  und nicht-deterministischen automaten. Elektronische Informationsverarbeitung
  und Kybernetik  12(6),  301--306 (1976)

\bibitem{Cer64}
Cern\'{y}, J.: Pozn\'{a}mka k. homog\'{e}nnym experimentom s konecnymi
  automatmi. In: Matematicko-fyzik\'{a}lny \v{c}asopis. vol. 14(3), pp.
  208--216 (1964)

\bibitem{CKVAK11}
Chadha, R., Korthikanti, V.A., Viswanathan, M., Agha, G., Kwon, Y.: Model
  checking {MDP}s with a unique compact invariant set of distributions. In:
  Proc. of QEST. pp. 121--130. IEEE Computer Society (2011)

\bibitem{CY95}
Courcoubetis, C., Yannakakis, M.: The complexity of probabilistic verification.
  J. ACM  42(4),  857--907 (1995)

\bibitem{DMS11b}
Doyen, L., Massart, T., Shirmohammadi, M.: Infinite synchronizing words for
  probabilistic automata. In: Proc. of MFCS. pp. 278--289. LNCS 6907, Springer
  (2011)

\bibitem{DMS11a}
Doyen, L., Massart, T., Shirmohammadi, M.: Synchronizing objectives for
  {M}arkov decision processes. In: Proc. of iWIGP: Interactions, Games and
  Protocols. pp. 61--75. EPTCS 50 (2011)

\bibitem{DMS11Err}
Doyen, L., Massart, T., Shirmohammadi, M.: Infinite synchronizing words for
  probabilistic automata ({E}rratum). CoRR  abs/1206.0995 (2012)

\bibitem{DMS14}
Doyen, L., Massart, T., Shirmohammadi, M.: Limit synchronization in {M}arkov
  decision processes. CoRR  abs/1310.2935 (2013), to appear in Proc. of FoSSaCS
  2014.

\bibitem{FV97}
Filar, J., Vrieze, K.: Competitive {Markov} Decision Processes. Springer (1997)

\bibitem{GB88}
Gibbons, A., Rytter, W.: Efficient parallel algorithms. Cambridge University
  Press (1988)

\bibitem{GO10}
Gimbert, H., Oualhadj, Y.: Probabilistic automata on finite words: Decidable
  and undecidable problems. In: Proc. of ICALP (2). pp. 527--538. LNCS 6199,
  Springer (2010)

\bibitem{HMW09}
Henzinger, T.A., Mateescu, M., Wolf, V.: Sliding window abstraction for
  infinite {M}arkov chains. In: Proc. of CAV. pp. 337--352. LNCS 5643, Springer
  (2009)

\bibitem{IS99}
Imreh, B., Steinby, M.: Directable nondeterministic automata. Acta Cybern.
  14(1),  105--115 (1999)

\bibitem{Kfo70}
Kfoury, D.: Synchronizing sequences for probabilistic automata. Studies in
  Applied Mathematics  29,  101--103 (1970)

\bibitem{KVAK10}
Korthikanti, V.A., Viswanathan, M., Agha, G., Kwon, Y.: Reasoning about {MDP}s
  as transformers of probability distributions. In: Proc. of QEST. pp.
  199--208. IEEE Computer Society (2010)

\bibitem{Martyugin14}
Martyugin, P.: Computational complexity of certain problems related to
  carefully synchronizing words for partial automata and directing words for
  nondeterministic automata. Theory Comput. Syst.  54(2),  293--304 (2014)

\bibitem{PK11}
Pinsky, M.A., Karlin, S.: An Introduction to Stochastic Modeling. Academic
  Press (2011)

\bibitem{Vardi-focs85}
Vardi, M.Y.: Automatic verification of probabilistic concurrent finite-state
  programs. In: Proc. of FOCS. pp. 327--338. IEEE Computer Society (1985)

\bibitem{Volkov08}
Volkov, M.V.: Synchronizing automata and the {C}erny conjecture. In: Proc. of
  LATA. pp. 11--27. LNCS 5196, Springer (2008)

\end{thebibliography}
\end{document}